\documentclass{imsart}
\usepackage{amsthm}
\usepackage{amsmath}
\usepackage{amssymb, amsfonts}
\usepackage{natbib}
\usepackage{bm}
\usepackage{bbm}
\usepackage{graphicx}
\usepackage{courier}
\usepackage{color}
\usepackage{floatrow}
\usepackage{placeins}
\usepackage{afterpage}
\usepackage{units}
\usepackage{url}
\usepackage{verbatim}
\usepackage{epstopdf}
\usepackage{booktabs,caption,fixltx2e}
\usepackage[flushleft]{threeparttable}
\usepackage{rotating}
\usepackage{multirow}
\usepackage{listings} 

\pubyear{2023}
\volume{0}
\issue{0}
\firstpage{1}
\lastpage{23}

\newcommand{\ind}{\stackrel{\mathrm{ind}}{\sim}}
\newcommand{\iid}{\stackrel{\mathrm{iid}}{\sim}}

\newtheorem{proposition}{Proposition}[section]

\usepackage{enumitem, hyperref, nameref}
\makeatletter 
\def\namedlabel#1#2{\begingroup
    #2%
    \def\@currentlabel{#2}%
    \phantomsection\label{#1}\endgroup
}
\makeatother

\begin{document}

\begin{frontmatter}
\title{Methods for Combining Probability and Nonprobability Samples Under Unknown Overlaps\protect\thanksref{T1}}
\runtitle{Methods for Combining Probability and Nonprobability Samples}
\runauthor{T. D. Savitsky et al.}
\thankstext{T1}{U.S. Bureau of Labor Statistics, 2 Massachusetts Ave. N.E, Washington, D.C. 20212 USA}

\begin{aug}
\author[A]{\fnms{Terrance D.}~\snm{Savitsky}\ead[label=e1]{Savitsky.Terrance@bls.gov}\orcid{0000-0003-1843-3106}},
\and
\author[B]{\fnms{Matthew R.}~\snm{Williams}\ead[label=e2]{mrwilliams@rti.org}}.
\and
\author[C]{\fnms{Julie}~\snm{Gershunskaya}\ead[label=e3]{Gershunskaya.Julie@bls.gov}}
\and
\author[A]{\fnms{Vladislav}~\snm{Beresovsky}\ead[label=e4]{Beresovsky.Vladislav@bls.gov}}
\and
\author[E]{\fnms{Nels G.}~\snm{Johnson}\ead[label=e5]{nels.johnson@usda.gov}}
\address[A]{Office of Survey Methods Research,
U.S. Bureau of Labor Statistics\printead[presep={,\ }]{e1,e4}}

\address[B]{RTI International\printead[presep={,\ }]{e2}}

\address[C]{OEUS Statistical Methods Division,
U.S. Bureau of Labor Statistics\printead[presep={,\ }]{e3}}


\address[E]{USDA Forest Service\printead[presep={,\ }]{e5}}
\end{aug}

\begin{abstract}
Nonprobability (convenience) samples are increasingly sought to reduce the estimation variance for one or more population variables of interest that are estimated using a randomized survey (reference) sample by increasing the effective sample size. Estimation of a population quantity derived from a convenience sample will typically result in bias since the distribution of variables of interest in the convenience sample is different from the population distribution.  A recent set of approaches estimates inclusion probabilities for convenience sample units by specifying reference sample-weighted pseudo likelihoods.  This paper introduces a novel approach that derives the propensity score for the observed sample as a function of inclusion probabilities for the reference and convenience samples as our main result.  Our approach allows specification of a likelihood directly for the observed sample as opposed to the approximate or pseudo likelihood.  We construct a Bayesian hierarchical formulation that simultaneously estimates sample propensity scores and the convenience sample inclusion probabilities.  We use a Monte Carlo simulation study to compare our likelihood based results with the pseudo likelihood based approaches considered in the literature.
\end{abstract}


\tableofcontents
\end{frontmatter}

\section{Introduction}

\subsection{Motivation}
With the proliferation of powerful computers and internet technologies, private data aggregators and research organizations gained the ability to relatively easily collect and store information from samples of respondents. Usually such opportunistic or ``convenience" samples are not selected using a probability based sampling design.
The non-random participation of units in such a convenience sample limits its ability to be used to construct an estimator (e.g., average income) 
of a target population quantity because the convenience sample, in general, is not expected to be representative of that population.

By contrast, probability based samples or random surveys of units represent the gold standard for cost-effectively sampling a population in a manner that allows provable guarantees about the population representativeness of target estimators (e.g., total employment, vaccination rate) composed from the observed sample where units are randomly invited to participate.  We term such a random-inclusions sample as a ``reference" sample.  
Yet, probability based samples are often relatively small, especially at finer domain levels; hence, probability based sample estimators often have large variances. In other cases,  reference samples may not include particular variables of interest, while such variables may be collected with the convenience sample.

Statistical agencies and other survey administrators are increasingly seeking ways to leverage convenience samples to construct estimators of target population quantities with measurable statistical properties.
This paper focuses on a class of approaches that suppose the nonrandom convenience sample was drawn from an unknown or latent random sampling design process such that we may treat the convenience sample as a ``pseudo" random sample.   The sampling design for the random reference sample is set by the governing statistical agency and is encoded in known sample inclusion probabilities assigned to the population of units.  These inclusion probabilities are used to form inverse probability sampling weights that are published with other variables collected 
in the reference sample. So, the task for combining the convenience sample with the reference sample to strengthen estimation (and lower the variance of estimators) is in \emph{estimation} of 
the unknown convenience sample inclusion probabilities to form ``pseudo" weights. We assume the existence of covariates, measured on both the reference and convenience samples, that encode the sampling design. Then, estimated convenience sample pseudo weights may be used with any response variable to form a weighted estimator of the target population quantity.


\subsection{Literature Review}


Early attempts to address estimation of the convenience sample inclusion probabilities using combined convenience and reference probability samples include \citet{2009elliot}, \citet{2011Dever}, \citet{2011DiSogra}. See recent reviews in \citet{2020Valliant} and \citet{2020Beaumont}, and \citet{2022Wu}.

Our goal in this paper is to estimate the convenience sample inclusion probabilities based on observed indicator $z_{i}$ that is defined on the combined convenience and probability samples set as $z_{i} = 1$ for a unit in the convenience sample, and $0$ for a unit in the reference sample.

\citet{2009elliot} and \citet{2017elliot} consider Bernoulli variable $z_{i}$ and uses relationship between $\pi_{zi}= P\{z_{i} = 1\}$, on the one hand, and the convenience and reference sample inclusion probabilities, $\pi_{ci}$ and $\pi_{ri}$ (respectively), on the other hand. One is then able to specify a logistic regression for estimation of $\pi_{ci}$.  While their result implies a practical approach, their derivation requires an assumption that the convenience and reference samples must be disjoint. That is, no unit may be included in \emph{both} the convenience and reference samples.   
They also use a two-step model estimation process that is suboptimal and often produces unbounded estimates for $\pi_{ci}$. A more efficient, one-step likelihood based estimation procedure, was proposed by \citet{beresovsky2019}.  

More recently, \citet{doi:10.1080/01621459.2019.1677241} approached the problem by considering the convenience sample inclusion indicator $R_{i}$, where $R_{i} = 1$ for unit $i$ in the convenience sample, and $0$ for unit $i$ in the finite population less those units which are members of the convenience sample. $R_{i}$ is a Bernoulli variate; however, convenience sample inclusion probabilities $\pi_{ci}=P\{ R_{i} = 1\}$ cannot be estimated directly from the Bernoulli likelihood of $R_{i}$ because the finite population is not generally available and indicator $R_{i}$ is not observed for the whole population; in particular, one does not know which units from the finite population are selected into the convenience sample. To overcome this difficulty, they partition the log-likelihood of $R_{i}$ into two terms: the sum over convenience sample units and the sum over the finite population. The latter term is approximated by a ‘’pseudo” likelihood, using inverse probability based weights, defined by observed reference sample inclusion probabilities. 

There are two shortcomings in \citet{doi:10.1080/01621459.2019.1677241}’s approach. 
First, the pseudo likelihood approximation is suboptimal because it is a noisy approximation on the observed sample that will produce a higher estimation variance. 
Second, convenience sample membership indicators $R_{i}$ are generally not observable. 
The partitioning proposed by \citet{doi:10.1080/01621459.2019.1677241} implies the existence of a different, \emph{observable}, indicator that is defined as follows. Stack together the convenience sample and finite population, so that the sample units appear in the stacked set twice: as part of the population and as the added set; let indicator $Z_{i} = 1$ for unit $i$ in the convenience sample, and $0$ for any unit $i$ in the finite population (regardless of whether it is also a part of the convenience sample). Note, however, that \citet{doi:10.1080/01621459.2019.1677241}’s likelihood does not treat observed $Z_{i}$ as a Bernoulli variate, thus potentially leading to suboptimal results.


\citet{2021valliant} propose an improvement of \citet{doi:10.1080/01621459.2019.1677241} by formulating the Bernoulli likelihood for $Z_{i}$ and providing a formula specifying a relationship between probabilities $P\{Z_{i} = 1\}$ and convenience sample inclusion probabilities $\pi_{ci}$. Would the finite population be observed, this approach would lead to efficient estimation of $\pi_{ci}$ based on the likelihood of observed $Z_{i}$. 
However, since the finite population is not observed, they still have to rely on the pseudo likelihood approach in their estimation. \citet{2021valliant} apply a two-step estimation procedure, which can be improved by solving the pseudo-likelihood based estimating equations using the one-step approach of \citet{beresovsky2019}.

\subsection{Contribution of this Paper}
We use first principles to derive a relationship between probability of being in the convenience sample set $\pi_{zi}$, on the one hand, and the convenience and reference sample inclusion probabilities, $\pi_{ci}$ and $\pi_{ri}$ (respectively), on the other hand. The result of \citet{2009elliot} can be viewed as a special case of our formula. Importantly, our approach dispenses with the requirement of disjointness between the two sample arms.  We show that our method for estimating $\pi_{ci}$ is valid under \emph{any} degree of overlapping units among the two sampling arms.  Unlike \citet{doi:10.1080/01621459.2019.1677241} and \citet{2021valliant}, our result is defined directly on the observed pooled sample with no approximation required. So, the resulting estimator of $\pi_{ci}$ from our method is more efficient than the approximate, pseudo likelihoods.

Differently from the two-step estimation process of \citet{2009elliot} or \citet{2021valliant}, we construct a Bayesian hierarchical modeling formulation discussed in the sequel that estimates both $(\pi_{zi},\pi_{ci})$ in a single step. Our method accounts for all sources of uncertainty to produce more accurate uncertainty quantification. 

Our approach is fully Bayesian for estimation of the unknown inclusion probabilities for the convenience sample units.   Notions of informativeness do not apply because the likelihood is formulated directly on the observed set. 
The model-estimated inclusion probabilities are subsequently used to compute sampling weights and those weights and the response variable are together used to construct a survey-based population estimator (such as the population mean of $y$).  

We introduce notation and list assumptions in Section~\ref{sec:prelims}. In Section~\ref{sec:bayesrule}, we detail the setup and provide the proof of the main formula underlying the proposed approach. Namely, we derive the relationship between the propensity score (defined as the probability of belonging to the convenience sample for a unit from the pooled sample), on the one hand, and the inclusion probabilities for the reference and convenience samples, on the other hand.  We construct a Bayesian hierarchical modeling formulation in Section~\ref{sec:model} that simultaneously estimates all unknown quantities, including unknown reference sample inclusion probabilities for convenience units, in a single step that accounts for all sources of uncertainty.   A Monte Carlo simulation study to compare our approach with competitor methods is presented in Section~\ref{sec:simulation}. In Section~\ref{sec:application}, we apply the proposed method to the Current Employment Statistics data, where we estimate pseudo weights for the non-probability based sample for local government in California and compute domain estimates based on these weights.  We conclude with a discussion in Section~\ref{sec:discussion}.

\section{Preliminaries}\label{sec:prelims}

We begin by introducing notation used in the exposition of our method developed in the following section. 
We follow by listing common assumptions used to develop the method.

Our set-up consists of a sample acquired under a random sampling design that we label as a ``reference" sample to contrast with availability of a nonrandom ``convenience" sample.  We term the observed pooled sample as a ``two-arm" sample with one arm denoting the reference (probability) sample and the other arm the convenience (nonprobability) sample. 

Let $S_c$ represent a non-probability (convenience) sample set drawn from sampling frame or population $U_c$, where $\vert U_{c}\vert = N_{c}$ and $\vert S_{c}\vert = n_{c}$ represent the number of units in sets $U_c$ and $S_c$, respectively; let $S_r$ denote a probability (reference) sample drawn from population $U_r$, with $\vert U_{r}\vert = N_{r}$ and $\vert S_{r}\vert = n_{r}$, the number of units, respectively, in $U_r$ and $S_r$. 

Let $U = U_{r} + U_{c}$ denote an imaginary combined set. Operator "+" here is meant to signify that sets $U_r$ and $U_c$ are ''stacked together" in such a way that overlapping units, that belong to both sets $U_r$ and $U_c$, would be included into $U$ twice. Similarly, let $S = S_{r} + S_{c}$ denote a pooled (stacked) sample. Under such a setup, $\vert U \vert=N_r+N_c=N$ and $\vert S \vert=n_r+n_c=n$.

In an abuse of notation, we index a unit contained in any population or observed sample realization by $i$, which may indicate a unit in any of the sample or population sets where the context is clear.

Let $\pi_{c}\left( \mathbf{x}_{i} \right)=P\left( i\in {{S}_{c}}|i\in U_c,\mathbf{x}_{i} \right)$ denote the probability of inclusion into observed sample set $S_c$ from $U_c$ conditional on associated design variables, $\mathbf{x}_{i}$.  We will use the term ``conditional inclusion probability" for an inclusion probability whose specification or estimation is conditioned on a set of design variables, $X = \{\mathbf{x}_{i}\}$.  These variables are used to construct the sampling design that governs the observed samples.  The design variables typically don't include one or more response variables, $y_{i}$, of inferential interest because they are not observed for the full underlying population (such their estimation motivates the administration of the survey).   

Let $\pi_{r}\left( \mathbf{x}_{i} \right)=P\left( i\in {{S}_{r}}|i\in U_r,\mathbf{x}_{i} \right)$ denote the conditional inclusion probability in $S_r$ from $U_r$.  


Let indicator variable $z_i$ on set $S$ take a value of $1$ when $i \in S_c$, and $0$ when $i \in S_r$; and let $\pi_z(\mathbf{x}_i)$ denote probabilities of $z_i=1$, given $\mathbf{x}_i$:  $\pi_z(\mathbf{x}_i)=P\left\{z_i=1|\mathbf{x}_i\right\}=P\left\{i \in S_c| i\in S,\mathbf{x}_i\right\}$.  We label $\pi_z(\mathbf{x}_i)$ as the ``propensity score" that measures the propensity or probability for a unit in the \emph{observed} joint sample, $S$, to be included in $S_{c}$.  

We will use $\pi_{ci}$ as a shorthand notation for $\pi_{c}(\mathbf{x}_{i})$ in the sequel when the context is clear and the same for $\pi_{ri}$.

\begin{description}
\item[\namedlabel{itm:latentrandom}{(C1)}] (Latent Random Mechanism)\\
     The observed convenience sample, $S_{c}$, is governed by an underlying, latent random mechanism with unknown sample inclusion probabilities, $\pi_{ci}$.
\item[\namedlabel{itm:design variables}{(C2)}] (Design Variables)\\
     $p \times 1$ variables, $X \in \mathcal{X}$, fully determine the unit conditional inclusion probabilities into $S_{r}$ and $S_{c}$ for the random selection mechanisms.  A consequence of the above set-up is that both $U_c$ and $U_r$ contain variables $\{X_{r},X_{c}\} \in \mathcal{X}$ on the same measure space.
\item[\namedlabel{itm:overlap}{(C3)}] (Overlapping Populations)\\
     Populations, $(U_{c}, U_{r})$, may overlap where units are jointly contained in each set such that overlapping units will each \emph{appear exactly twice} in $U$. As a result, observed samples $(S_{c}, S_{r})$ may also contain overlapping units such that overlapping units each appear twice in $S$.  
\item[\namedlabel{itm:independence}{(C4)}] (Independence of Samples)\\
    Conditional on $X$, $S_{r} \perp S_{c} \mid X$.  Inclusions of units into each sample arm are independent, no matter the degree of overlap between $U_{r}$ and $U_{c}$.  
\item[\namedlabel{itm:positivity}{(C5)}] (Positive Inclusion Probabilities)\\
    For all $i \in 1,\ldots,n$ and for all $\mathbf{x} \in \mathcal{X}$, conditional inclusion probabilities in each sampling arm are strictly positive / non-zero, such that $P(i \in S_{r} \mid \mathbf{x}_{i}, i \in U) > 0,~ P(i \in S_{c} \mid \mathbf{x}_{i},i \in U) > 0$, which leads to $P(i \in S \mid \mathbf{x}_{i}, U) > 0$.  These conditions result in $P(i \in S \mid i \in U) = \int P(i \in S \mid \mathbf{x}, i \in U)F(d\mathbf{x}) > 0.$
\end{description}


Assumption~\ref{itm:latentrandom} states that the non-random convenience sample may be understood as governed by a latent random process that we seek to uncover. The focus of this paper is the estimation of unknown inclusion probabilities into the convenience sample. 

Convenience sample inclusion probabilities, $\pi_{ci}$, are generally not observed for units in the convenience sample; e.g., $\forall i \in S_{c}$.  Reference sample inclusion probabilities are not generally observed for those units sampled solely into the convenience sample (and not included in the reference sample); e.g., $\forall j \in S_{c}\setminus S_{r}$.

Assumption~\ref{itm:overlap} allows for a general case of non-perfectly overlapping convenience and reference frames (from which the associated two samples are taken). 

Our method requires Assumption~\ref{itm:independence} on the independence of the reference and inclusion samples, but makes no assumptions about the degree of unit overlaps between the two samples.

It is typical to assume positive inclusion probabilities for all units as we do in Assumption~\ref{itm:positivity} for any rational sampling design in order to ensure that every unit in population $U$ may be sampled, which in turn allows for unbiased inference about the population for the observed samples taken under this assumption.

\section{Likelihood Based Estimation of Inclusion Probabilities Under Two-arm Samples}\label{sec:bayesrule}




In this section we prove an identity that is central to our proposed approach for estimation of convenience sample inclusion probabilities. The proof is made from the first principals and  under no requirement for disjointness among the sample arms. Namely, we derive the relationship between the propensity for the observed set of reference and convenience inclusion indicators and the associated inclusion probabilities in each sample.

Suppose, each frame is a subset of target population $U^0$, such that $U_c \subseteq U^0$ and $U_r \subseteq U^0$. 
Define probabilities $p_c\left( \mathbf{x}_i \right)=P\left\{i\in U_c|i \in U^0,\mathbf{x}_i\right\}$ and $p_r\left( \mathbf{x}_i \right)=P\left\{i\in U_r|i \in U^0,\mathbf{x}_i\right\}.$ Quantities $p_{r}(\mathbf{x}_{i})$ and $p_{c}(\mathbf{x}_{i})$ are \emph{known} coverage probabilities of population $U^0$ by frames $U_c$ and $U_r$ for a set of design variables $\mathbf{x}_i$.  These probabilities depend on the same design variables, $\mathbf{x}_{i}$, though units will express differing values for the common design variables.  For example, frame $U_c$ could be the subset of individuals in $U^0$ with broadband internet access and frame $U_r$ could be the subset of individuals in $U^0$ with mailable addresses. 

While conditional inclusion probabilities $\pi_{r}\left( \mathbf{x}_i \right)=P\left\{ i\in {{S}_{r}}|i\in U_r,\mathbf{x}_i \right\}$ for sample $S_{r}$ are known,
convenience sample conditional inclusion probabilities $\pi_{c}\left( \mathbf{x}_i \right)=P\left\{ i\in S_c|i\in U_c,\mathbf{x}_i \right\}$ are unknown and can be inferred from combined sample $S=S_c+S_r$, where samples $S_c$ and $S_r$ are stacked together.  As already mentioned in previous sections, samples $S_c$ and $S_r$ \emph{may overlap}. The overlapping units appear in (stacked) set $S$ \emph{twice}: as units from $S_c$ (with $z_i=1$) and as units from $S_r$ (with $z_i=0$). \\
\\
\emph{Proposition}: Assume Conditions \ref{itm:latentrandom}-\ref{itm:positivity}. Then, the following relationship between respective probabilities holds:
\begin{align}\label{eq:pi_z}
\begin{split}
{{\pi }_{z}}\left( \mathbf{x}_i \right)=\frac{{{\pi }_{c}}\left( \mathbf{x}_i \right)p_{c}\left( \mathbf{x}_i \right)}{{{\pi }_{c}}\left( \mathbf{x}_i \right)p_{c}\left( \mathbf{x}_i \right)+{{\pi }_{r}}\left( \mathbf{x}_i \right)p_{r}\left( \mathbf{x}_i \right)}.
\end{split}
\end{align}
\emph{Proof}: The combined set $S$ emerges from the following scheme displayed in Figure \ref{fig:scheme} where we stack identical populations $U^0$ of units.  The set of units in $U^0$ are \emph{duplicated} from the top-to-the-bottom stack. In the top layer we define $U_{r}$ from which we draw sample, $S_{r}$ and we do the same for the convenience population, $U_{c}$, and sample, $S_{c}$, in the bottom stack.  We see that units in $U_{r}$ and $U_{c}$ may overlap in this scheme, which allows units in $S_{r}$ and $S_{c}$ to also overlap, though we don't know the identities of overlapping units because we have duplicated them in each stack, so our notation separately indexes the same unit in the reference and convenience frames and observed samples.  This means that the sampling processes in each stack are independent from one another, but readily permit overlaps in $(U_{r},U_{c})$ and $(S_{r},S_{c})$.  We next outline the scheme of Figure \ref{fig:scheme} in our proof.

To summarize, we consider two copies of target population $U^0$, where one copy of the population includes frame $U_c$, the other copy includes $U_r$. We \emph{stack} the two copies of $U^0$ together and denote the result by $U$: $U=U^0+U^0$.

For such a setup, by the Law of Total Probability (LTP), we have:
\begin{align}\label{eq:S_c}
\begin{split}
P\left\{ i\in S_c|i \in U,\mathbf{x}_i \right\}
&=P\left\{ i\in S_c|i\in U_c,i \in U^0,\mathbf{x}_i \right\}P\left\{i\in U_c|i \in U^0,\mathbf{x}_i\right\}P\{i \in U^{0} \mid i \in U\}\\
&=\frac{1}{2}\pi_{c}\left( \mathbf{x}_i \right)p_c\left( \mathbf{x}_i \right)
\end{split}
\end{align}
We note that $i \in S_{c}$ implies that $i \in U_{c}$ since we draw the convenience sample from its associated frame, $U_{c}$. Similarly,
\begin{align}\label{eq:S_r}
\begin{split}
P\left\{ i\in S_r|i \in U,\mathbf{x}_i \right\}
&=P\left\{ i\in S_r|i\in U_r,i \in U^0,\mathbf{x}_i \right\}P\left\{i\in U_r|i \in U^0,\mathbf{x}_i\right\}P\{i \in U^{0} \mid i \in U\}\\
&=\frac{1}{2}\pi_{r}\left( \mathbf{x}_i \right)p_r\left( \mathbf{x}_i \right).
\end{split}
\end{align}
Now, because we have stacked $U^0$ twice - once for the convenience sampling process and again for the reference sampling process - thus ''shifted" sets $S_c$ and $S_r$ do not overlap (as illustrated in Figure \ref{fig:scheme}), so we may sum them below to compute the total probability of being included into the pooled sample,
\begin{align}\label{eq:S}
\begin{split}
P\left\{ i\in S|i \in U,\mathbf{x}_i \right\}
&=P\left\{ i\in S_c|i \in U,\mathbf{x}_i \right\}+P\left\{ i\in S_r|i \in U,\mathbf{x}_i \right\}\\
&=\frac{1}{2}\pi_c\left( \mathbf{x}_i \right)p_c\left( \mathbf{x}_i \right)+\frac{1}{2}\pi_{r}\left( \mathbf{x}_i \right)p_r\left( \mathbf{x}_i \right).
\end{split}
\end{align}
Finally, by the definition of conditional probability, 
\begin{align}\label{eq:S_c|S}
P\left\{ i\in S_c|i\in S,i \in U,\mathbf{x}_i \right\} &= \frac{P\left\{ i\in S_c|i \in U,\mathbf{x}_i \right\}}{P\left\{ i\in S|i \in U,\mathbf{x}_i \right\}}.
\end{align}
Equation \ref{eq:pi_z} directly follows from Equations \ref{eq:S_c}, \ref{eq:S}, and \ref{eq:S_c|S}. \\
\begin{figure}
  \captionsetup{singlelinecheck = false, 
  format= hang, 
  justification=justified, font=footnotesize, labelsep=space}
\centering
  \includegraphics[width=0.7\linewidth]{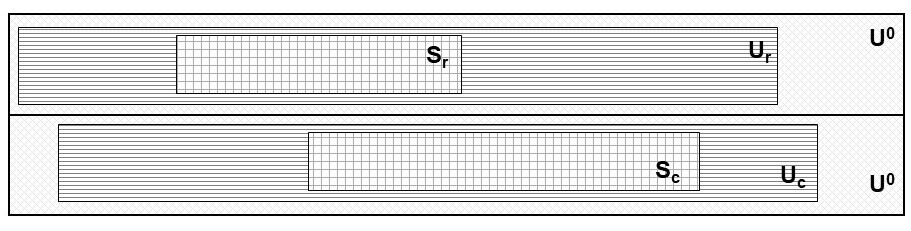}
 \caption{Gridded area represents observed convenience $S_c$ and reference $S_r$ samples stacked together to form combined sample $S$: $S=S_c+S_r$; under this scheme, if samples $S_c$ and $S_r$ overlap, the overlapping units are included in $S$ twice. Convenience sample $S_c$ is selected from population $U_c$, and reference sample $S_r$ is selected from population $U_r$, where $U_c$ and $U_r$ are subsets of target population $U^0$: $U_c \subseteq U^0$ and $U_r \subseteq U^0$. In this setup, two identical copies of target population $U^0$ are stacked together, so that $U=U^0+U^0$.}
  \label{fig:scheme}
\end{figure}
\\
We may now parameterize a likelihood for the observed indicator $z_i$ using Equation~\ref{eq:pi_z}: 
\begin{align}\label{eq:Bernoulli}
z_i\mid \mathbf{x}_{i},\bm{\beta} \ind \mbox{Bernoulli}(\pi_z(\mathbf{x}_i,\bm{\beta})).
\end{align}
The likelihood of Equation \ref{eq:Bernoulli} implicitly depends on parameter $\pi_c(\mathbf{x}_i)$ through Equation \ref{eq:pi_z}. 
We may specify a model for $\pi_c(\mathbf{x}_i) = f(\mathbf{x}_i, \bm{\beta})$ and fit parameters using either Frequentist or Bayesian approaches. We use a Bayesian approach in the sequel for its flexibility.\\
\\
\emph{Remark 1:}  Our formulation for the propensity score does not rely on disjointness among the sampling arms. Our method explicitly allows for the \emph{unknown} overlapping of units in $S_{r}$ and $S_{c}$.

\noindent\emph{Remark 2}: We can view the process as a two-phase selection. First, units are selected from target population $U^0$ to subpopulations $U_c$ and $U_r$ with probabilities $p_{c}(\mathbf{x}_i)$ and $p_{r}(\mathbf{x}_i)$, respectively. At the second phase, units are selected to respective samples with probabilities $\pi_{c}(\mathbf{x}_i)$ and $\pi_{r}(\mathbf{x}_i)$.\\

\noindent\emph{Remark 3: The equal frame scenario.} If frames $U_c$ and $U_r$ coincide, we have $p_{c}\left( \mathbf{x}_i \right)=p_{r}\left( \mathbf{x}_i \right)$, and Equation~\ref{eq:pi_z} becomes 
\begin{equation}\label{eq:U_c=U_r}
{{\pi }_{z}}\left( \mathbf{x}_i \right)=\frac{{{\pi }_{c}}\left( \mathbf{x}_i \right)}{{{\pi }_{c}}\left( \mathbf{x}_i \right)+{{\pi }_{r}}\left( \mathbf{x}_i \right)}.
\end{equation}
\\
A similar expression was derived by \cite{2009elliot} under the assumption of non-overlapping convenience and reference samples. Our approach does not require this assumption.
 
 
Equation~\ref{eq:U_c=U_r} holds even when the reference sample is the entire target population frame $U$. In this case, $\pi_r\left( \mathbf{x}_i \right)=1$ for all units and Equation~\ref{eq:U_c=U_r} reduces to  
\begin{equation}
{{\pi }_{z}}\left( \mathbf{x}_i \right)=\frac{{\pi }_{c}\left( \mathbf{x}_i \right)}{{{\pi }_{c}}\left( \mathbf{x}_i \right)+1},
\end{equation}
which is the same as that of \citet{2021valliant} before they approximate it on the observed sample. 
We label this as a ``one-arm" case. One important simplification in the ``one-arm" case is that $\pi_r$'s are known (and equal to 1) for all units in combined set $S$.\\ 

We derive the same result as presented in Equation~\ref{eq:pi_z} under perfectly overlapping frames by extending a different result from the economics literature in  Appendix~\ref{sec:symmetric}.

\section{Hierarchical Estimation Model}\label{sec:model}
We next specify a hierarchical probability model to estimate convenience sample inclusion probabilities for the units in the convenience sample.  

We focus on the equal frame scenario where both the reference and convenience samples are assumed to be drawn from the same underlying frame to define our Bayesian hierarchical model and simulation setup. We do so for ease and clarity of explanation, with no loss of generality. In the common case where the frames do not perfectly overlap, we would use Equation \ref{eq:pi_z} which inputs coverage probabilities $p_c\left( \mathbf{x}_i \right)$ and $p_r\left( \mathbf{x}_i \right)$ as \textit{known} quantities.

We assume that our covariates $\mathbf{x}$ fully account for the sampling design. Thus, our goal is to formulate a model to estimate the inclusion probabilities of convenience sample units given covariates $\mathbf{x}$. We use them to formulate inverse probabilities based pseudo sampling weights to construct a survey expansion estimator using response variable of interest $y$. 

\subsection{Construction of unknown marginal inclusion probabilities, \texorpdfstring{$(\pi_{\ell i})$}{Lg}}
We parameterize our model using $\pi_{\ell i} = P\left\{i \in S_{\ell}\mid i \in U_{\ell},\mathbf{x}_{i}\right\}$ to be the conditional inclusion probability for unit $i \in 1,\ldots,\left(n = n_r + n_c\right)$ in sampling arm $\ell \in (r,c)$; that is, $\ell$ indexes whether the conditional inclusion probability for unit $i$ is specified for the reference ($\ell=r$) or the convenience ($\ell = c$) sampling arms.  This modeling set-up only assumes that we observe $\pi_{\ell i}$ for $\ell = r$ and $i \in S_{r}$, the conditional sampling inclusion probabilities for the units observed in the reference sample.

Our model, however, will estimate $(\pi_{\ell i})$ for \emph{all} units, $i \in (1,\ldots,n)$, for both $ \ell = r$ and $\ell = c$ sampling arms.  Of particular note, our model estimates $\pi_{ri}$ for $i \in S_{c}$, the reference sample inclusion probabilities for the convenience units.  So, estimation of the model does not require known $\pi_{ri}$ for all units.  The model will further simultaneously estimate $\pi_{ci}$ for $i \in S_{c}$, the convenience sample inclusion probabilities for the convenience units (units in the convenience sampling arm), which is the primary goal of the model.

A Bayesian hierarchical model is able to be richly parameterized to estimate this matrix of only partially observed conditional inclusion probabilities through the borrowing of strength in the specifications of functional forms and prior distributions to follow.

\subsection{Spline functional form for $\mbox{logit}(\pi_{\ell i})$}
We input an $n \times K$ matrix of design variables, $X = (\mathbf{x}_{1},\ldots,\mathbf{x}_{K})$, where $\mathbf{x}_{k}$ denotes the $n \times 1$ vector for design variable $k$.  We want our model specification to express a flexible functional form,
\begin{equation}
    \mbox{logit}(\pi_{\ell i}) = f(x_{1i},\ldots,x_{Ki}),
\end{equation}
where $f(\cdot)$ may be estimated as non-linear by the data.  Complex sampling designs may utilize design variables with different emphases on different portions of the design space, which will induce such non-linearity.  Two common examples are (i) scaling inclusion probabilities to exactly meet target sample sizes and (ii) thresholding size measures to create certainty units (with $\pi_{\ell i} = 1$). Both features induce non-linearity on the logit scale.

To accomplish the above non-linear formulation we utilize a B-spline basis due to its flexibility and computational tractability for illustration (of an implementation of our main result in Equation~\ref{eq:pi_z}).  We may also choose alternative non-linear formulations, such as a Gaussian process, to achieve similar results, but computation for  the Gaussian process scales poorly in the number of data observations.  The use of Bayesian adaptive regression trees is not easily purposed to our modeling set-up for estimating latent convenience sample inclusion probabilities \citep{Chipman_2010}.

A B-spline basis is specified for \emph{each} predictor where $Q \times 1$, $g(x_{ki})$ is a B-spline basis vector with $C$ denoting the number of bases set equal to the number of knots + number of spline degrees - 1.  We use the vector of B-splines for each predictor $k$ to formulate,
\begin{equation}\label{eq:ff}
    \mbox{logit}(\pi_{\ell i}) = \mu_{x,\ell i} = \mathbf{x}_{i}^{\top}\bm{\gamma}_{x,\ell} + \mathop{\sum}_{k=1}^{K} g(x_{ki})^{\top}\bm{\beta}_{\ell k},
\end{equation}
where $\mathbf{x}_{i}^{\top}\bm{\gamma}_{x,\ell}$ is a linear component and $\bm{\beta}_{\ell k}$ is a $~Q\times 1$ vector of coefficients for the spline term for each predictor $k$ (column of $X$) that parameterizes the possibility for a non-linear functional form for each of the $K$ predictors.  The spline term specifies distinct regression coefficients for each sampling arm, $\ell$, and design variable, $k$, to allow estimation flexibility that makes few assumptions about the functional form for $\mbox{logit}(\pi_{\ell i})$.  In this sense, even if we had only used the linear term, the use of distinct spline term regression coefficients for each predictor and sampling arm makes the model marginally non-linear across the data.
\subsection{Random walk of order 1 (autoregressive) horseshoe prior on $\beta_{\ell k}$}
We select a random walk of order 1 (based on first differences) formulation for the prior on each component of the $Q\times 1$, $\bm{\beta}_{\ell k}$ of the spline term with,
\begin{equation}
    \beta_{\ell k q} \mid  \beta_{\ell k q-1}, \kappa_{\ell k}\tau_{\ell} \sim \mathcal{N}\left(\beta_{\ell k q-1}, \kappa_{\ell k}\tau_{\ell}\right),~ c = 2,\ldots,Q,
\end{equation}
and $\beta_{\ell k 1}  \sim \mathcal{N}\left(0, \kappa_{\ell k}\tau_{\ell}\right)$ denotes a spline basis (used for each predictor $k \in 1,\ldots,K$). All to say, the random walk prior is constructed for the B-spline coefficients defined on each predictor, $x_{k}$. This random walk form for the prior enforces smoothness over the estimated regression coefficients such that the resulting estimated fit is less sensitive to the number of (spline) knots used and avoids overfitting.   The overall mean intercept is identified  by excluding an intercept from the linear term in Equation~\ref{eq:ff}.

We also encourage sparsity in the number of estimated non-zero, $(\bm{\beta}_{\ell k})_{k=1}^{K}$, as a group for predictor $K$, by using a set of $K$ ``local" scale (standard deviation) shrinkage parameters, $\kappa_{\ell k}$, where a value for $\kappa_{\ell k^{'}}$ near $0$ for some predictor $k^{'}$ will shrink all $Q \times 1$ coefficients, $\bm{\beta}_{\ell k^{'}}$, to $0$. Similarly, global scale (standard deviation) shrinkage parameter, $\tau_{\ell}$, would shrink \emph{all} $(\bm{\beta}_{\ell k})_{k=1}^{K}$ to $0$, which favors the linear model term in this limit.   We place half Cauchy priors, $\kappa_{\ell k} \ind C^{+}(0,1)$ and $\tau_{\ell} \sim C^{+}(0,1)$, respectively.  

This use of local and global shrinkage parameters under a half Cauchy prior is known as the horseshoe prior \citep{Carvalho_handlingsparsity}.  If one marginalizes out the global and local scale shrinkage parameters under the half Cauchy priors, the marginal prior distribution for $\beta_{\ell k q}$ will have a large spike at 0 (driving sparsity), but with very heavy tails allowing the coefficient values to ``escape" the shrinkage where the data provide support.   By tying together the priors for $(\beta_{\ell k q})_{q=1}^{Q}$ the spline coefficients for predictor $k$ escaping shrinkage to $0$ will be correlated and relatively smooth.

The vector of $K\times 1$ fixed effects parameters for sampling arm $\ell$ are each drawn as,
\begin{equation}
    \gamma_{x,\ell k} \iid \mathcal{N}\left(0,\tau_{\gamma}\right),
\end{equation}
where $\tau_\gamma \sim \mbox{student-t}^{+}(\mbox{d.f.} = 3, 0, 1)$, where we use a relatively flat prior for $\tau_{\gamma}$.

\subsection{Joint likelihood for $(z_{i})_{i \in S}$ and $({\pi}_{ri})_{i\in S_{r}}$}
We connect our parameters to the data with two likelihood terms.  The first term constructs a Bernoulli likelihood for the observed sample,
\begin{equation}\label{bernlike}
    z_{i} \mid \pi_{zi} \ind \mbox{Bernoulli}(\pi_{zi}),
\end{equation}
where we recall from Equation~\ref{eq:U_c=U_r} that, $\pi_{zi} = \pi_{ci} / (\pi_{ci} + \pi_{ri})$ such that this likelihood provides information for estimation of $\pi_{ci}$ for $i \in 1,\ldots,n$, as well as $\pi_{ri}$ for $i \in S_{c}$. 
 
We further borrow strength from the known reference sample conditional inclusion probabilities for the observed reference sample to estimate the unknown conditional inclusion probabilities by modeling the known reference sample inclusion probabilities for the observed reference sample units as a function of our parameters with,
\begin{equation}
\label{eq:likepir}
    \mbox{logit}({\pi}_{ri}) \ind \mathcal{N}(\mu_{x,ri},\phi),
\end{equation}
only for units $i \in S_{r}$ such that observed ${\pi}_{ri}$ is used to provide information about latent $\mu_{x,\ell i}$ for both sampling arms $(\ell \in (r,c)$ and all units ($i \in 1,\ldots,n$) based on their intercorrelations allowed by the prior distribution (and updated by the data).  We recall from Equation~\ref{eq:ff} that each $\mu_{x,\ell i}$ is, in turn, connected with each $\pi_{\ell i}$. 

The detailed Stan \citep{gelman2015a} script that enumerates the likelihood and prior distributions for all parameters and hyper-parameters is included in Appendix~\ref{a:script}.

\subsection{Bayesian hierarchical model implementations for pseudo likelihoods}

We implement the pseudo likelihood formulations of \citet{doi:10.1080/01621459.2019.1677241} and \citet{2021valliant} in the simulation study of Section~\ref{sec:simulation} under our Bayesian hierarchical formulation.   We accomplish these implementations by replacing the exact Bernoulli likelihood for the observed sample under our method of Section~\ref{sec:bayesrule} with approximate likelihoods for the underlying population estimated on the sample.   Both methods parameterize only $\pi_{ci}$ for convenience units and use the inclusion probabilities for the reference sample as a plug-in.   Let vector $\bm{\theta}_{c} = \left(\bm{\gamma}_{x,c},(\bm{\beta}_{ck})_{k=1}^{K}\right)$ denote the parameters in the non-linear logistic regression model for $\pi_{ci}(\mathbf{x}_{i},\bm{\theta}_{c})$.  \citet{doi:10.1080/01621459.2019.1677241} specify the following pseudo log-likelihood, 

\begin{equation}\label{eq:wulike}
    \ell(\bm{\theta}_{c}) = \mathop{\sum}_{i \in S_{c}}\log\left(\frac{\pi_{ci}(\mathbf{x}_{i},\bm{\theta}_{c})}{1-\pi_{ci}(\mathbf{x}_{i},\bm{\theta}_{c})}\right) + \mathop{\sum}_{i \in S_{r}}d_{ri}\log\left(1-\pi_{ci}\left(\mathbf{x}_{i},\bm{\theta}_{c}\right)\right),
\end{equation}
where $d_{r} = 1/\pi_{ri}$.  
Equation~\ref{eq:wulike} uses a survey approximation for the population in the second term by inverse probability weighting the reference sample contribution.   This pseudo likelihood will tend to produce overly optimistic (narrow) credibility intervals because it uses $\pi_{ri}$ as a plug-in (rather than co-modeling it).  The first term will also induce a noisy estimator for unit with low values for $\pi_{ci}$, which will occur when there is a lot of separation in the covariate, $\mathbf{x}$, values between the convenience and reference samples.

As discussed in the introduction, \citet{2021valliant} develop a Bernoulli likelihood for the population augmented by the convenience sample.  This approach specifies indicator $Z_{i} = 1$ if unit $i$ is in the convenience sample, or $0$ if it is the finite population and develops an associated propensity score, $\pi_{Zi} = \pi_{ci}/(\pi_{ci} + 1)$. This expression is a special case of our formula derived in Section~\ref{sec:bayesrule} where one arm is the convenience sample and the other arm is the entire population. So the exact likelihood specified in \citet{2021valliant} is a special case of our method under a one-arm sample set-up.  As with \citet{doi:10.1080/01621459.2019.1677241}, they approximate their log-likelihood on the observed sample with
\begin{equation}
    \ell(\bm{\theta}_{c}) = \mathop{\sum}_{i \in S_{c}}\log\left(\pi_{Zi}(\mathbf{x}_{i},\bm{\theta}_{c})\right) + \mathop{\sum}_{i \in S_{r}}d_{ri}\log\left(1-\pi_{Zi}(\mathbf{x}_{i},\bm{\theta}_{c})\right).
\end{equation}
This approximate likelihood will also tend to produce overly optimistic credibility intervals because it doesn't account for the uncertainty in the generation of samples (by plugging in the reference sample weights, $d_{ri}$, instead of modeling them). 

Both comparator methods are implemented under our hierarchical Bayesian model such that they are benefited from our flexible, nonlinear formulation for the logit of the convenience sample inclusion probabilities and our autoregressive smoothing on spline coefficients.  In this sense, these implementations are more robust than the estimating equation approaches used by the authors. In addition, in our implementation of \citet{2021valliant}, we estimate convenience sample probabilities in a single step.    


\section{Simulation Study}\label{sec:simulation}
We construct a finite population and two sets each of reference and convenience samples characterized by low and high overlaps in number of overlapping units between the two sampling arms.  We perform this construction in each iteration of a Monte Carlo simulation study designed to compare the repeated sample (frequentist) properties of our two-arm exact likelihood approach with those of the pseudo likelihood approaches.  We compare bias, root mean squared error and coverage of $90\%$ credibility intervals. 

\subsection{Simulation Settings}
To compare performance variability across multiple realized populations, we generate $M = 30$ distinct populations of size $N = 4000$. We chose a relatively small population size and large sampling fractions to explore the full range of $\pi_c \in [0,1]$. A large sampling fraction and large inclusions probabilities is also reasonable in establishment surveys. We set the reference sample size at $n_{r} = 400$ using a proportion-to-size (PPS) sampling. We select convenience samples of size $n_{c} \approx 800$ using Poisson sampling.  We recall our assumption that the convenience sample arises from a latent random sampling mechanism with unknown inclusion probabilities. We select two distinct independent convenience samples from each population, which we deem `high' and `low' overlap in comparison to the reference sample. High-overlap convenience samples have selection probabilities $\pi_c$ with a similar relationship with population covariates $X$ compared to the selection probabilities $\pi_r$ for the reference sample. In constrast, low-overlap convenience sample probabilities $\pi_c$ have the opposite relationship with covariates.

For each population, we let $X$ have $K = 5$ columns, including an intercept, three independent binary variables (A,B,C) with $P(\mathbf{x}_i = 1) = 0.5$, and a continuous predictor drawn from a standard normal distribution $N(0,1)$. We do not explore the situation of correlated design variables in this simulation study. We generate the outcome $y_i$ as a lognormal distribution with centrality parameter $\mu_i = \mathbf{x}_i \beta$ and scale parameter 2: $\log(y_i) \sim \mathcal{N}(\mu_i, 2)$. The generating parameters are $(\beta_{cont}, \beta_0, \beta_A, \beta_B, \beta_C) = (1.0, 0.5, 0.0, -0.5, -1.0)$. The inclusion probabilities for the reference sample are constructed by first setting size measure $s_{r_i} = \log(\exp(\mu_i) + 1)$. We then convert size to inclusion probabilities $\pi_{r_i}$ via the \emph{inclusionprobabilities()} function from the `sampling' package in R \citep{samplingR}. Most sizes of $s_{r_i} \propto \pi_{r_i}$, however the largest size values get mapped to $\pi_{r_i} = 1$, thus inducing a non-linear `kink' in the mapping from $s_{r_i} \rightarrow \pi_{r_i}$. We note that our estimation model $\mbox{logit}(\pi_{r_i}) = \mu_i$ is misspecified leading to a non-linear relationship with the $x_i$. This motivates the use of splines to capture non-linear relationships and add robustness to the model estimation.  It is common for the largest-sized units to be included in the sample with probability $1$.  

The inclusion probabilities for the convenience samples are inverse logit transformations of linear predictors with an offset adjustment to the intercept to approximately meet a target sample size: $\pi_{c_i} = \mbox{logit}^{-1}(\mathbf{x}_i \bm{\beta} + \mbox{off})$. For the high-overlap sample:
\newline $(\beta_{cont}, \beta_0, \beta_A, \beta_B, \beta_C, \mbox{off}) = (0.500, 0.175, -0.150, -0.475, -0.800, -0.900)$. For the low overlap sample: $(\beta_{cont}, \beta_0, \beta_A, \beta_B, \beta_C, \mbox{off}) = (-1.00, -0.50, 0.00, 0.50, 1.00, -2.23)$.
It is generally more challenging to estimate convenience sample inclusion probabilities when there is a lower overlap of predictor values with the reference sample.

\begin{figure}
\centering
\includegraphics[height = 0.48\textwidth,
		page = 1,clip = true, trim = 0.0in 0.0in 0.in 0.in]{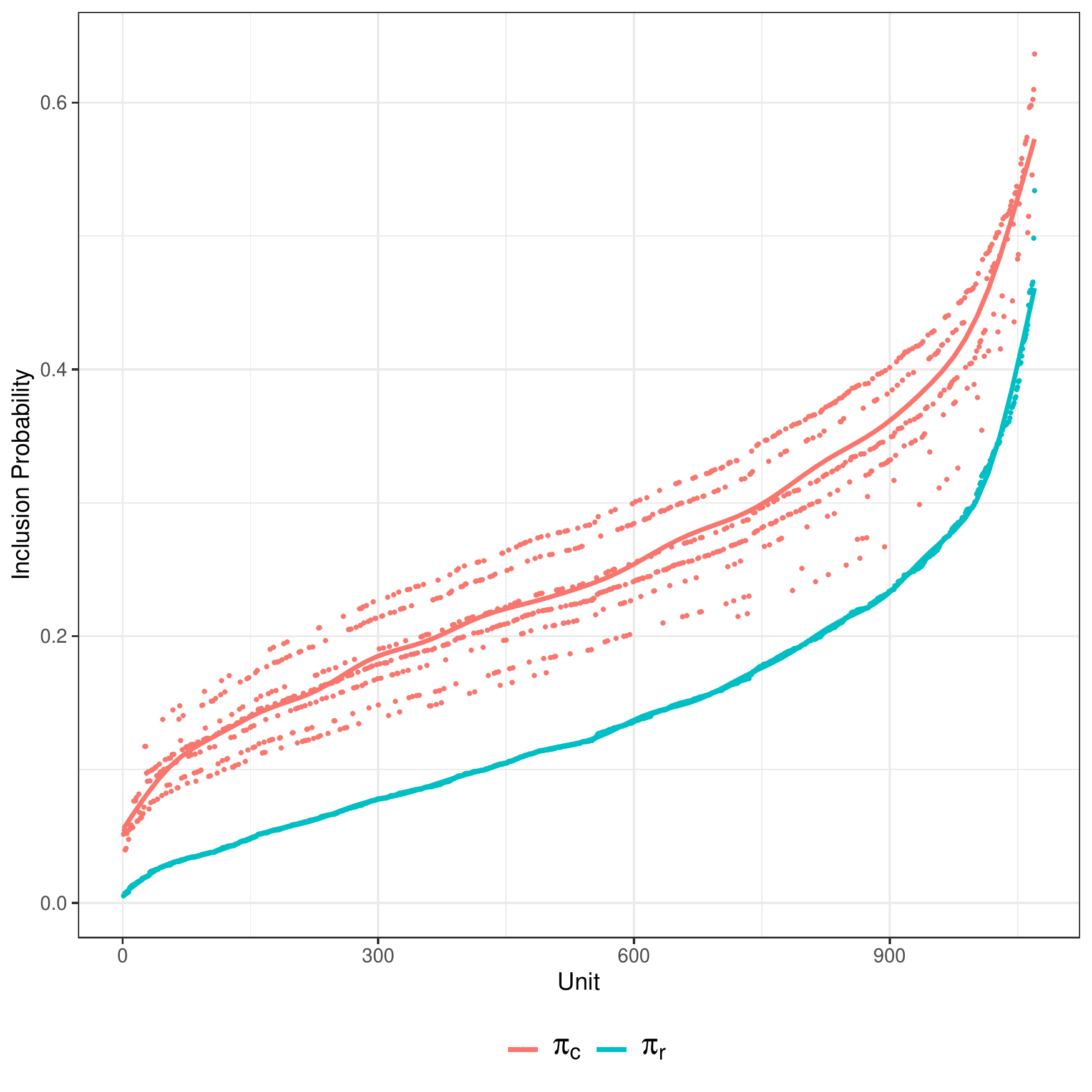}
		\includegraphics[height  = 0.48\textwidth,
		page = 2,clip = true, trim = 0.0in 0.0in 0.in 0.in]{rev_overlap_hi_lo_with_infsample_x_4000_param_5_sim_1.pdf}
\caption{Comparison of inclusion probabilities for a single realization of reference  and convenience samples for high overlap (left) and low overlap (right) designs. Units index the combined sample.
    }
\label{fig:sample_overlap}
\end{figure}

\begin{figure}
\centering
\includegraphics[height = 0.48\textwidth,
		page = 1,clip = true, trim = 0.0in 0.0in 0.in 0.in]{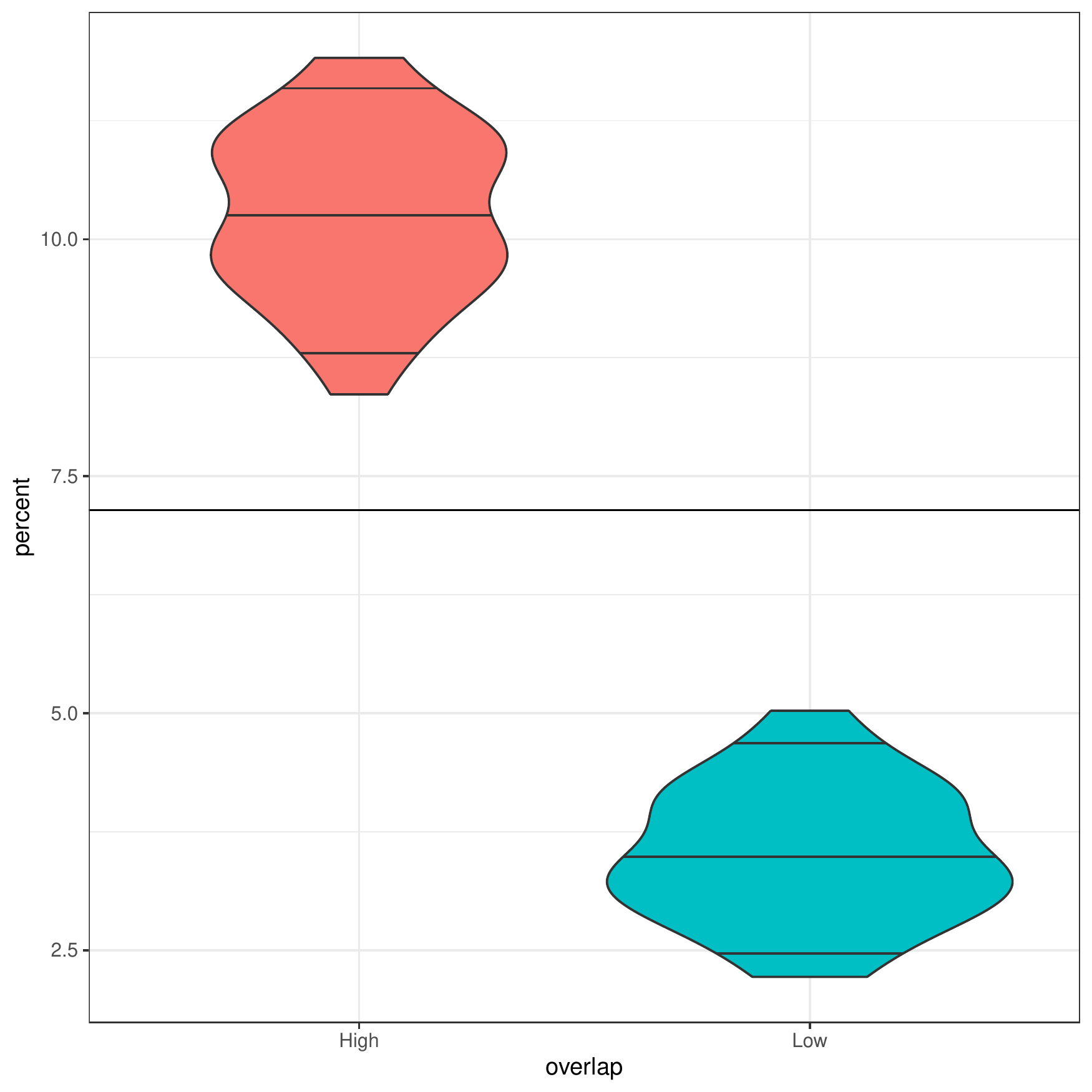}
\caption{Percent of pooled sample present in both reference and convenience samples by type of convenience sample (High and Low). Distributions over 30 population and sample realizations. Expected percent for two independent simple random samples (solid horizontal line).
    }
\label{fig:sample_overlap_count}
\end{figure}

Each plot panel in Figure~\ref{fig:sample_overlap} compares the generated reference sample inclusion probabilities to the convenience sample inclusion probabilities for under a high-overlap size-based sampling design on the left and a low-overlap sampling design on the right.  Each plot panel orders units by reference sample inclusion probabilities low-to-high along the x-axis.  The degrees of similarity in the reference and convenience sample inclusion probabilities are achieved by manipulating the size and direction of the vector of coefficients $\bm{\beta}$ for the design variables.   

Figure \ref{fig:sample_overlap_count} compares the percent of the total combined sample (reference and convenience) units which overlap (e.g. is present in both samples) for realizations of `high' and `low' overlap convenience samples as well as a baseline expected overlap from two independent simple random samples. As indicated by their labels, `high' overlap samples have a larger proportion of individuals in both the reference and convenience sample than each of a sample of the same sizes based on SRS and a `low' overlap sample.

\subsection{Results - Estimating convenience sample inclusion probabilities, $\pi_{ci}$}

We begin our presentation of results by comparing the relative performances of exact (two-arm) and pseudo likelihood methods for the estimation of the convenience sample inclusion probabilities, $\pi_{ci}~i \in S_{c}$ based on our known true values.

The plot panels of Figure~\ref{fig:main_hilo} present the mean of bias (over the Monte Carlo iterations), the square root of the mean squared error and the (frequentist) coverage and average widths of $90\%$ credibility intervals from left-to-right in the matrix of plot panels.  These values are computed pointwise for increasing values of the true conditional inclusion probabilities from left-to-right in each plot panel.  The top row of plot panels presents results for high-overlap (convenience and reference) sample datasets and the bottom row presents results for low-overlap sample datasets.   


We compare $3$ methods:
\begin{enumerate}
    \item two-arm - constructs an exact likelihood from our main method of Equation~\ref{eq:U_c=U_r} under a two-arm convenience and reference sample set-up.
   \item CLW - The pseudo likelihood method of \citet{doi:10.1080/01621459.2019.1677241}.
   \item WVL - The pseudo likelihood method of \citet{2021valliant}.
\end{enumerate}
For the two pseudo likehood methods, we implement each directly as pseudo posteriors and with a post-processing adjustment using a sandwich estimate of an asymptotic covariance matrix \citep{williams2021}. The stability of estimation of the sandwich estimator for CLW was poor. In order to compensate, we first used a scalar down-weighting (or tempering) of all observations (both convenience and reference) such that the sum of the sum of the individual weights was equal to the total sample size. See \cite{10.1214/18-AOS1712} for a detailed discussion on the stabilization of posterior estimation using such fractional weights.

We see that our two-arm method 
produces little mean bias for both small and large values of the true convenience sample inclusion probabilities and 
achieves nominal coverage of the $90\%$ credibility intervals.  

By contrast, both pseudo likelihood methods perform similarly to one another with high variability (RMSE) and severe undercoverage for medium-to-larger values of true convenience sample inclusion probabilities $\pi_c$. The collapse in coverage becomes worse for the low overlap dataset as the use of reference sample weights as a plug-in both under-estimates the uncertainty introduced by the reference sample design and induces noise over repeated samples. For high overlap, a post-processing adjustment for the pseudo likelihood methods improves coverage at the expense of increasing the width of the corresponding interval beyond that of the two-arm method. For low overlap, the post-processing adjustment can only adjust variance but not bias. In fact it may even amplify bias. Coverage is improved, but at the cost of very wide intervals.

\begin{figure}
\centering
\includegraphics[width = 0.95\textwidth,
		page = 1,clip = true, trim = 0.0in 0.0in 0.in 0.45in]{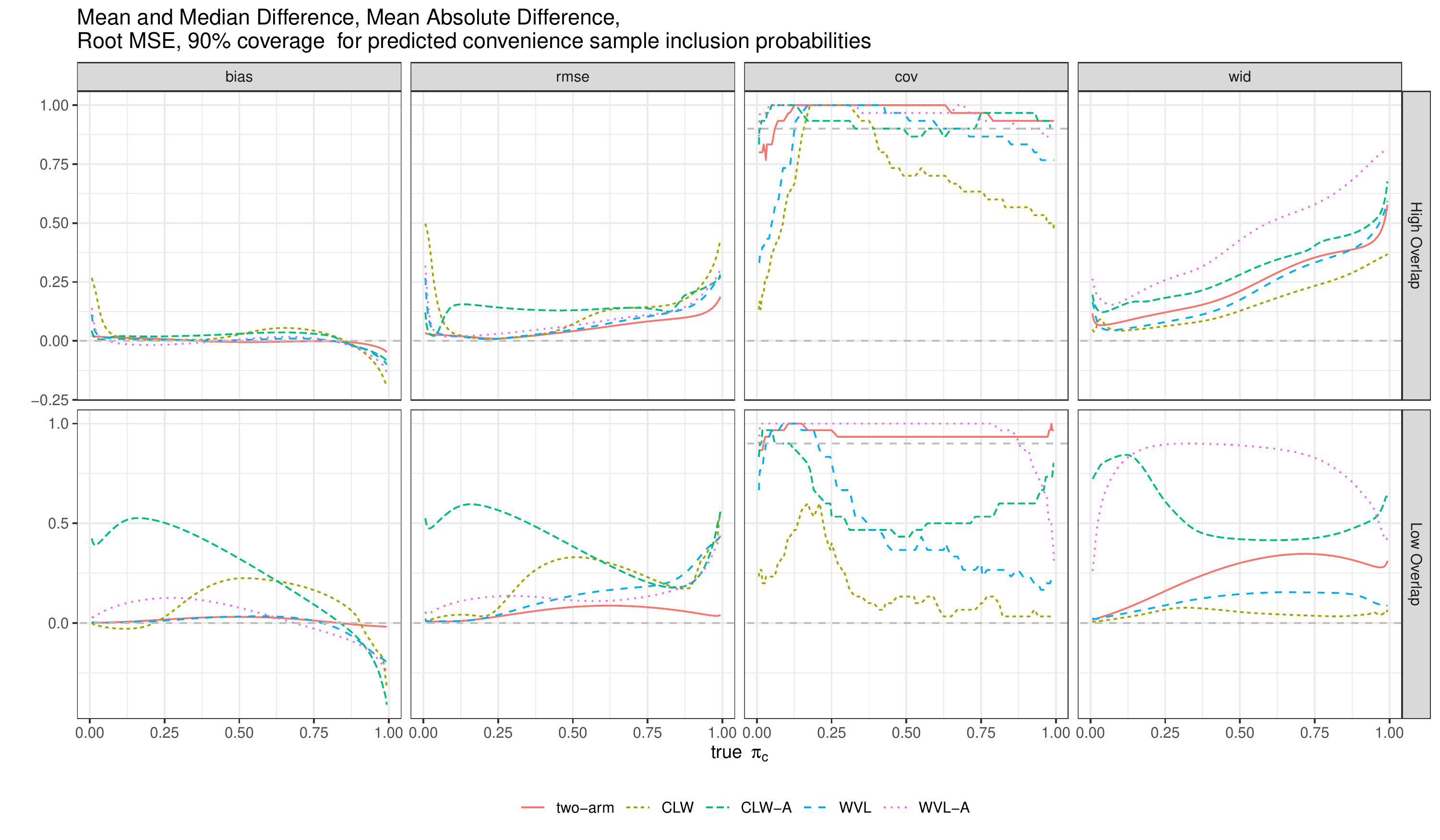}

\caption{Performance of main approach for high overlap (top) and low overlap (bottom) samples across repeated simulations. Using informative reference sample for main approach (red), compare to pseudo-likehood based methods CLW (yellow) and WVL (blue). Adjusted versions of pseudo-likehood adjust based on an estimated sandwich covariance matrix: CLW-A (green) and WVL-A (purple). Left to Right: Mean Bias, Square Root Mean Squared Error, and Coverage and Interval Width for 90\% intervals for predicting convenience sample inclusion probabilities $\pi_c$
    }
\label{fig:main_hilo}
\end{figure}

Each plot panel in Figure~\ref{fig:main_hilo_curve} compares the average and pointwise $95\%$ frequentist confidence intervals for the posterior mean estimator of $\pi_c$ over the Monte Carlo iterations.  The left-hand panel represents the results for the high-overlap datasets and the right-hand panel for the low-overlap datasets.  We see that the 
two-arm exact likelihood method produces little-to-no bias. 
 By contrast, the pseudo likelihood methods produce an enormous amount of variability. 

While we would expect the performance of the pseudo likelihood methods to improve as the sample sizes increases since both methods produce consistent estimators, our chosen sample size is a very typical domain sample size for a survey such that the superior performance of our exact likelihood methods at these moderate sample sizes (for ($n_{r}, n_{c})$) is an important result that demonstrates much faster convergence for our approach.

\begin{figure}
\centering
\includegraphics[width = 0.95\textwidth,
		page = 2,clip = true, trim = 0.0in 0.0in 0.in 0.45in]{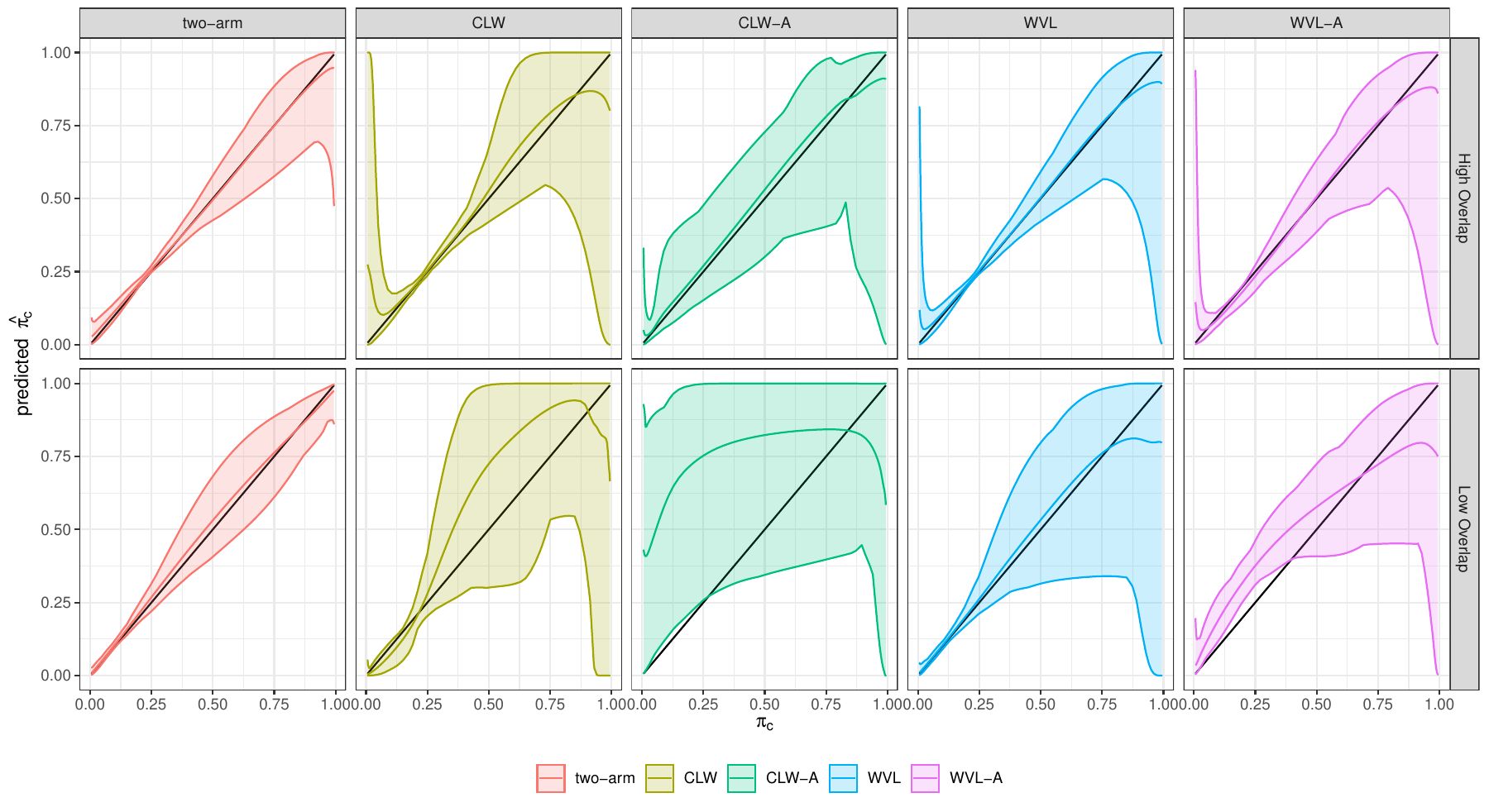}

\caption{Average and pointwise $95\%$ frequentist confidence intervals for the posterior mean estimator of $\pi_c$ over the Monte Carlo iterations for high overlap (top) and low overlap (bottom) samples. Using informative reference sample for main approach (red), compare to pseudo-likehood based methods CLW (yellow) and WVL (blue). Adjusted versions of pseudo-likehood adjust based on an estimated sandwich covariance matrix: CLW-A (green) and WVL-A (purple).
    }
\label{fig:main_hilo_curve}
\end{figure}

We use our method to combine convenience and reference sample inclusion probabilities $(\pi_{ci},\pi_{ri})$ to construct a non-model-based survey direct estimator for the population mean, $\mu$, of some response variable of interest, $y$, that is correlated with the survey design variables, $\mathbf{x}$ in Appendix~\ref{a:mu}.  We compare our resulting population mean estimator to that estimated from the two pseudo likelihood methods.

\section{Application}\label{sec:application}

We next present results from applying our proposed method to estimate pseudo weights for a quota sample of government employment collected in the Current Employment Statistics (CES) survey administered by the U.S. Bureau of Labor Statistics (BLS). We subsequently apply the pseudo weights to estimate local government employment for the Metropolitan Statistical Areas (MSA) of California.

The CES uses probability-based sampling design for private industries. For government employers, however, the CES estimates are based on a non-probability sample. The employment coverage in government industries is generally high, so that the resulting unweighted estimates based on such a non-probability sample usually provide acceptable level of precision. For measuring employment of local governments, however, such an unweighted quota (convenience) sample based estimate may be biased. 

We will use the quarterly census of employment and wages (QCEW), which is a census instrument administered by BLS that measures establishment employment, as our ``reference sample" to estimate the pseudo weights for the CES government convenience sample. As a large census instrument, QCEW quality checking and reporting are lagged by many months, so the CES is used to provide the current month employment. The QCEW employment levels are maintained in an administrative source called the longitudinal database (LDB).

To estimate pseudo weights, we stack together the LDB and the CES sample and apply Equation~\ref{eq:pi_z} that links the propensity score for the pooled sample, $\pi_{zi}$, to the convenience (CES) inclusion probabilities, $\pi_{ci}$, and the reference sample inclusion probabilities, $\pi_{ri}$. 

The LDB is designed to cover the target population; therefore, we set ${\pi}_r=1$ for all units in LDB, regardless of $\mathbf{x}_i$. In addition, coverage probabilities are set $p_r=1$ and $p_c=1$ for all units.  In the case that LDB frame were insufficient and didn't cover all of the CES sample we could set $p_{r} < 1$ in our set-up to account for it.  We observe: $z=1$ for units in the CES sample and $z=0$ for units in the LDB. Note, even though CES units are a subset of LDB, we are not concerned with matching the CES to LDB. Instead, we stack the two sets together. Thus, CES units appear in the stacked set \emph{twice}: once with $z=1$ and again with $z=0$.

We apply our model to estimate probabilities ${\pi }_c\left( \mathbf{x}_i \right)$ of inclusion into the CES sample, where $\mathbf{x}_i$ is employment level of unit $i$ in September (the benchmark month). We formulate our model with domain level random effects $u_d$ and use splines as described in Section \ref{sec:model}.

The fit performance is assessed by comparing CES based estimates to QCEW-based employment levels that become available to researchers on a lagged basis. Due to different seasonality patterns between the employment series derived from QCEW data and CES, the most meaningful comparison of the two series is after $12$ months of estimation. Mimicking the production setup, we obtain level estimates after $12$ months of estimation from monthly ratio estimates, $\hat{R}_{d,\tau}$, for a set of domains $d \in 1,\ldots,N$ at month $\tau$. The monthly ratio estimates are multiplied together and by the September starting level, $Y_{d,0}^{{}}$, that is available to CES at the start of the estimation cycle,
\begin{equation*}
\hat{Y}_{d,12}^{{}}=Y_{d,0}^{{}}\prod\limits_{\tau =1}^{12}{{{{\hat{R}}}_{d,\tau }}}.
\end{equation*}
Monthly ratio estimates $\hat{R}_{d,\tau}$ are obtained using a link relative (LR) estimator, that is a ratio of the sum of the current month to the sum of previous month responses, over set $s_{d,\tau}$ of CES respondents at a given month $\tau$ in domain $d$: $\hat{R}_{d,\tau}^{LR}=\sum_{i \in s_{d,\tau}}{y_{i,\tau}}/\sum_{i \in s_{d,\tau}}{y_{i,\tau-1}}$. Once we apply our approach to obtain pseudo weights $w_i$, we use them in the analogous formula to form a ''pseudo" \emph{weighted} link relative (WLR) estimator, $\hat{R}_{d,\tau}^{WLR}=\sum_{i \in s_{d,\tau}}{w_{i}y_{i,\tau}}/\sum_{i \in s_{d,\tau}}{w_{i}y_{i,\tau-1}}$.

We extract the posterior means of the pseudo weights and apply them to each month in the estimation cycle. Figure \ref{fig:annualcycle_msa} displays examples of estimates of employment levels over the 12 months of the estimation cycle for California MSAs under both the LR and WLR estimators, both compared to the QCEW Historical (Hist) truth (that we obtain on a lagged basis).  We readily see that our pseudo weighted WLR estimator generally does a better job of estimating the truth.

Figure \ref{fig:annualrev} shows the distribution of \emph{annual revisions} of the level estimates based on LR and WLR methods, respectively, over the set of MSAs in California. The annual revision, ${rev}_{d,12}$, is defined as the difference between the respective estimate, $\hat{Y}_{d,12}$, and ``true" population level ${Y}_{d,12}$ that becomes available after the fact, at the 12th month after the benchmark month:
\begin{equation*}
{rev}_{d,12}=\hat{Y}_{d,12}-{Y}_{d,12}.
\end{equation*}
Again, the WLR estimator demonstrates better fit performance than does the LR estimator in that the distribution of revision magnitudes is more compact.


To compute variance $v_{d,\tau}^{WLR}$ of the WLR estimate of relative change ${R}_{d,\tau}$, we extract $10$ draws from the posterior distribution of the fitted pseudo weights, estimate sampling variance for each draw of the pseudo-weights and then use a multiple imputation procedure described in a Appendix~\ref{a:mu} to compute the total variance of ${R}_{d,\tau}$ in a manner that accounts for the uncertainty in the estimation of weights. Coefficients of variations, $cv_{d,\tau}=\sqrt{v_{d,\tau}^{WLR}}/{R}_{d,\tau},$ are presented in Figure \ref{fig:cv_scatter}, where they are plotted against the employment level of respective domains. It can be observed that variances tend to be smaller in larger domains, as is expected.

\begin{figure}
\includegraphics[width = 0.95\textwidth]{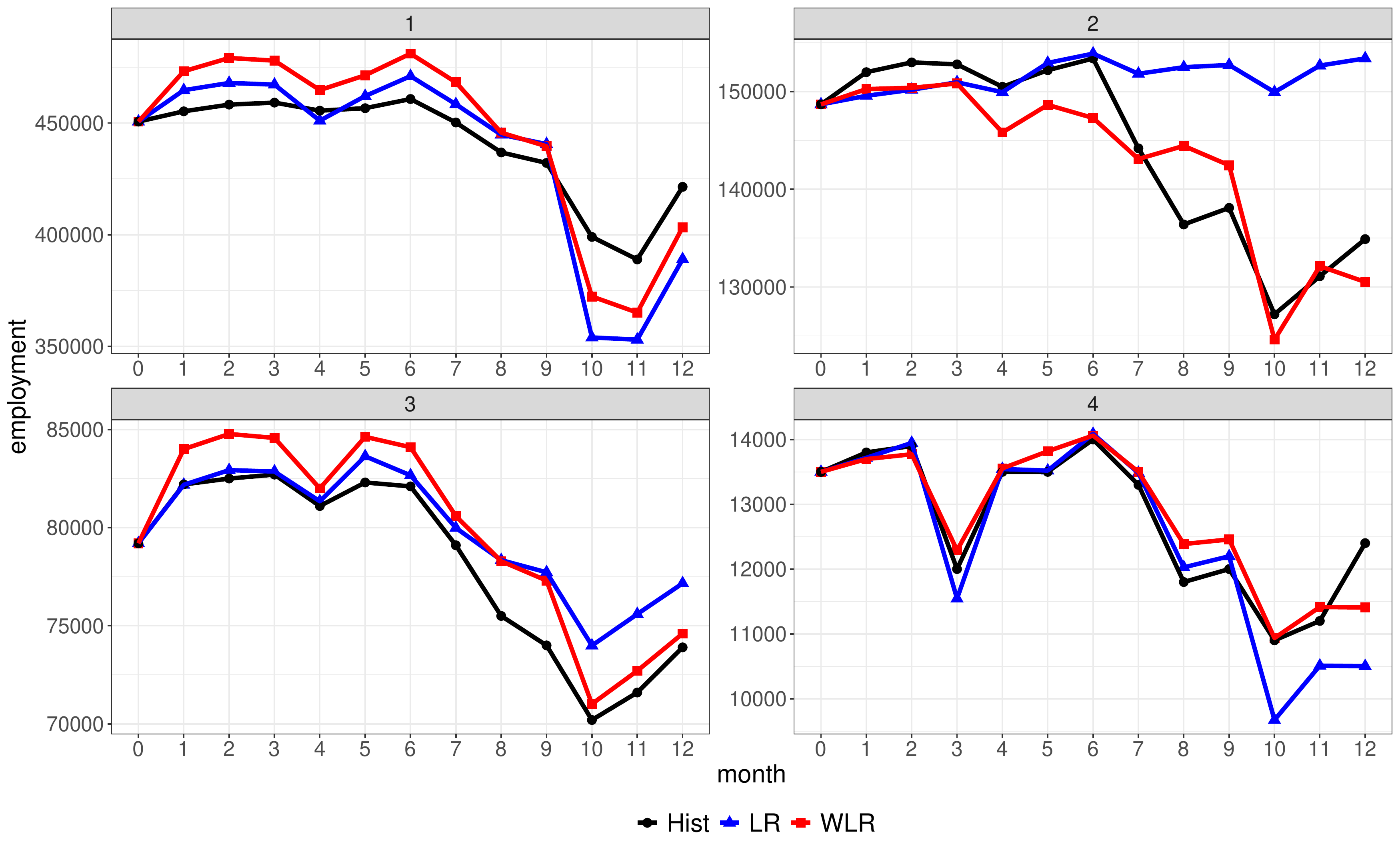}
\caption{Examples of a 12-month CES estimation cycle for select MSAs in California, series starting from September 2019 true levels. The black solid line corresponds to the "true" monthly levels ("Hist", obtained after the fact from historical series), the blue line with triangles shows estimates based on unweighted monthly link relatives ("LR"), and the red line with squares shows estimates based on the weighted link relatives ("WLR").}
  \label{fig:annualcycle_msa}
\end{figure}

\begin{figure}
\includegraphics[width = 0.95\textwidth]{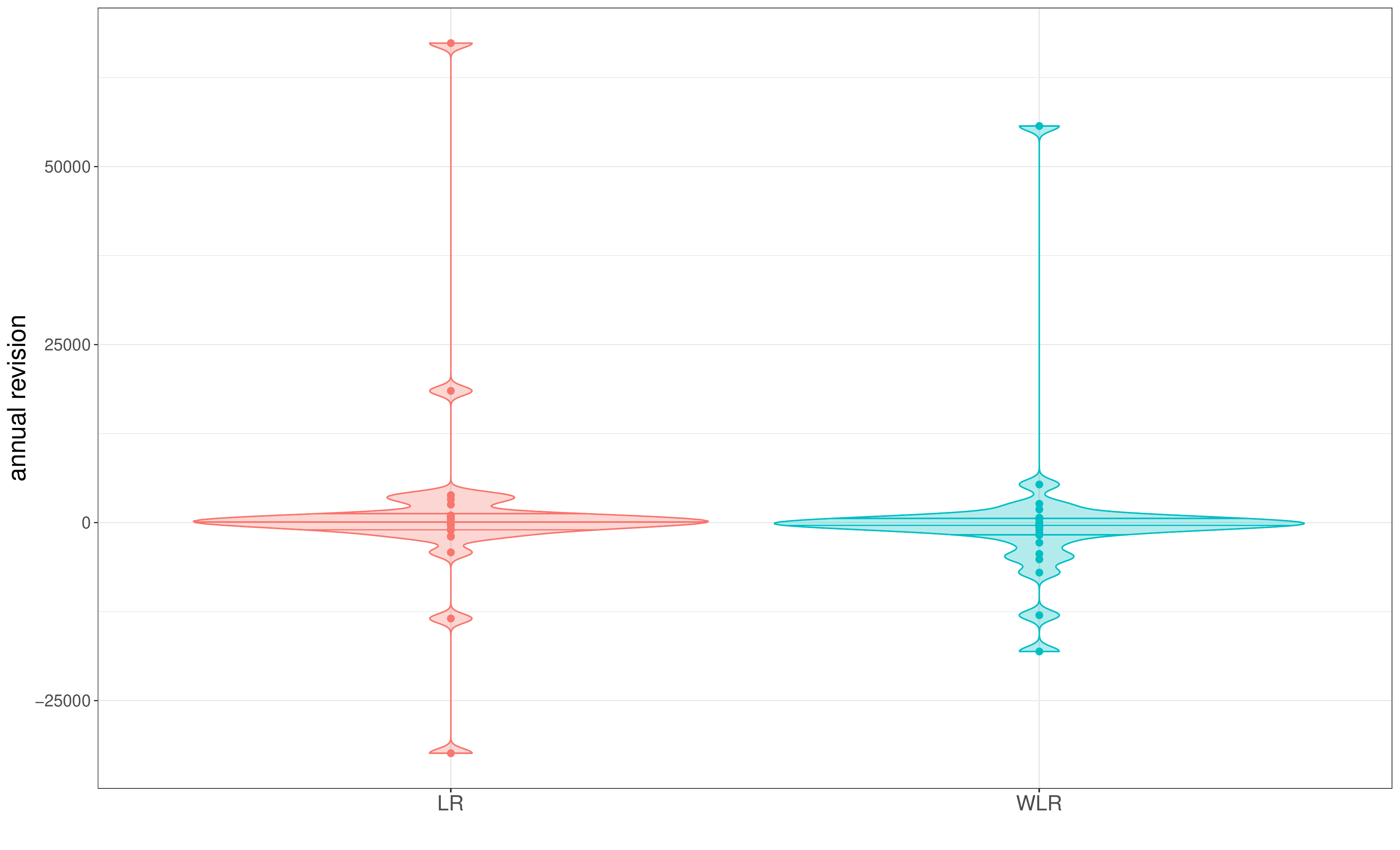}
\caption{Distribution of annual revisions for MSAs in California. "Annual revisions" are differences between respective level estimates (LR or WLR based) and the true historical levels at the 12th month of the estimation cycle.}
\label{fig:annualrev}
\end{figure}

\begin{figure}
\includegraphics[width = 0.95\textwidth]{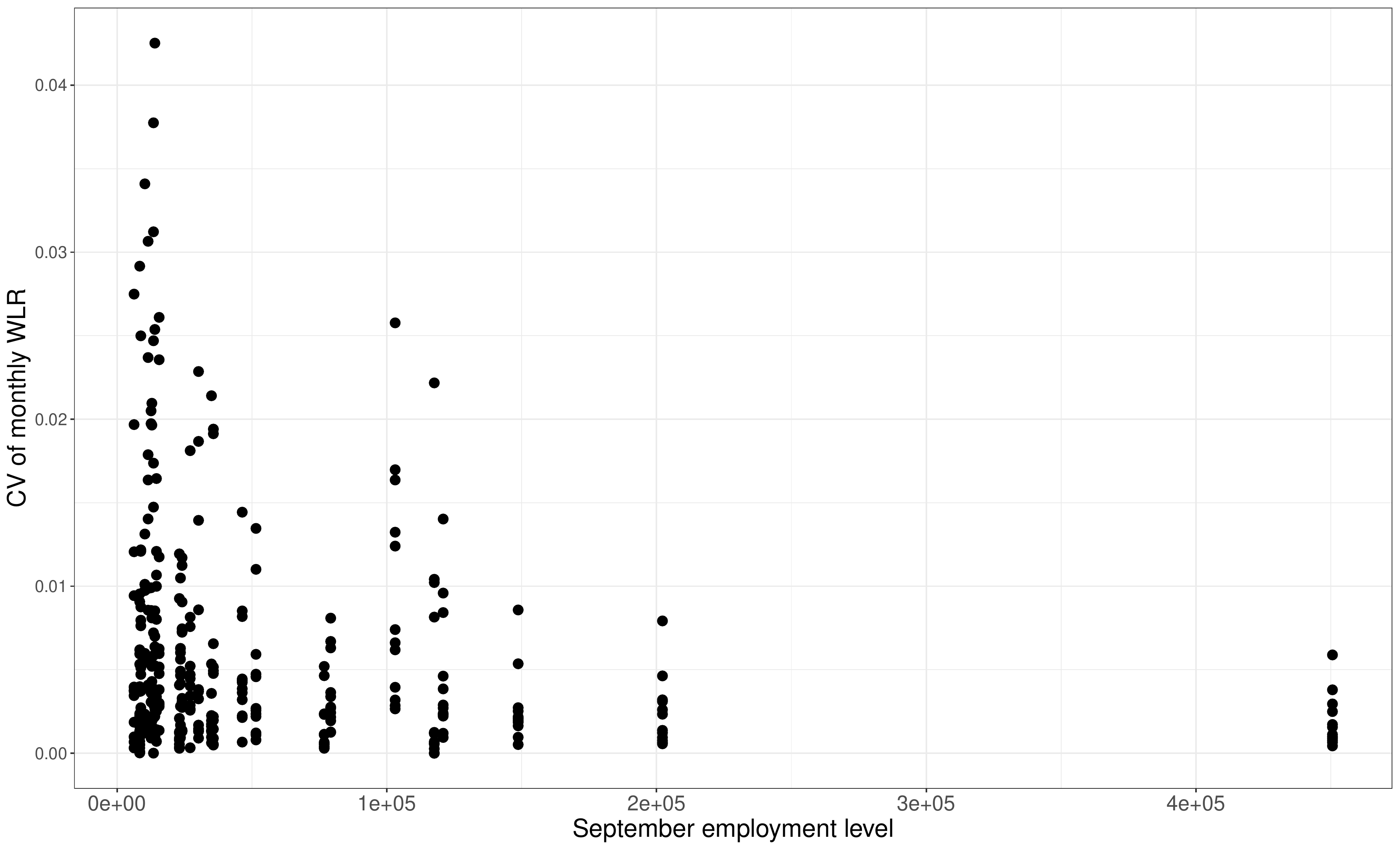}
\caption{Coefficients of variation (CV) of monthly WLR estimates for 12 months of estimation versus respective domains employment levels at the start of the estimation period in September 2019, for California MSAs.}
\label{fig:cv_scatter}
\end{figure}

\section{Discussion}\label{sec:discussion}
We introduced a novel approach that derived an exact relationship between the sample propensity score, $\pi_{zi}$, on the one hand, and the reference and convenience samples conditional inclusion probabilities, $\pi_{ri}$, and $\pi_{ci}$, on the other hand for an observed pooled sample.  Our expression is valid for any size of the overlap between the reference and convenience samples. It allows us to specify a likelihood directly for the sample using $\pi_{zi}$ and our specification of a Bayesian hierarchical probability model to simultaneously estimate all of them.

\appendix

\section{Estimation of Inclusion Probabilities Under Symmetric Two-arm Sampling} \label{sec:symmetric}

Our main method derives an expression connecting $(\pi_{zi},\pi_{ci},\pi_{ri})$ on the observed sample from first principles using the survey sampling literature.  We proceed on an alternative path that also connects these quantities based on the economics literature.  We will see in the sequel that this alternate path produces the same estimator, though they are derived from completely different approaches.

\citet{lancaster1996case} provide a modeling formulation for estimation of the conditional sample inclusion probabilities for case observations of interest under a two-arm experimental design with one-arm consisting of cases and the other consisting of an unknown collection of both case observations and control observations.  Under an observed sample from each arm, they assume a $2-$step sample observation process where the first step is a Bernoulli draw for the observed sub-sample indicator into either a case sample arm or a mixed case and control sample arm, given the observed sample.  The second step consists of the realization or appearance of units from the selected arm in the first stage.  The process parameterizes an exact likelihood for the distribution for predictors, $\mathbf{x}$, conditioned on the sub-population of cases in that sampling arm and a marginal population distribution for $\mathbf{x}$ in the mixed arm.  Using the distribution for $\mathbf{x}$ allows a clever and simple specification of the marginal distribution for $\mathbf{x}$ since they don't know the mix of cases and controls in the second arm.  The conditional distribution in the case sampling arm is a function of the case sample conditional inclusion probability (by Bayes rule) parameterized by regression coefficients.   This approach has the virtue of simultaneous estimation of conditional propensity scores and the conditional inclusion probabilities for cases.

We proceed to specialize and extend their $2-$step sample observation process and use of conditional distributions for $\mathbf{x}$ to our set-up of reference and control sampling arms and will specify a likelihood in each arm based on the sub-population of units linked to each type of sample.   

Let $z_{i} \in \{0,1\}$ be the same binary inclusion indicator of selection into the convenience sample for unit $i \in (1,\ldots,n)$ used in the previous section.  When $z_{i} = 0$ unit $i$ is drawn from the reference sample.   We suppose the observed two-arm sample (with convenience and reference sample arms) arises from a Bernoulli draw into either arm with probability $P(z_{i} = 1) = P(i \in S_{c}\mid i \in S)$ and subsequently specify a conditional sub-population distribution for $\mathbf{x}_{i}$ whose form depends on the outcome of the Bernoulli draw for each unit, $i \in (1,\ldots,n)$.  In particular, $p(\mathbf{x}_{i} \mid i \in S_{c}) = \pi_{c}(\mathbf{x}_{i}\mid\bm{\beta}_{c}) \times f(\mathbf{x}_{i}) / P(i \in S_{c} \mid i \in S, i \in U_{c})$ for the convenience sample by Bayes rule where we recall that $\pi_{c}(\mathbf{x}_{i}\mid\bm{\beta}_{c}) = P(i\in S_{c} | \mathbf{x}_{i}, \bm{\beta}_{c})$.  We drop the conditioning on $U_{c}$ and $U_{r}$ in the sequel where the context is obvious for readability.  We have included regression parameters $\bm{\beta}_{c}$ that parameterizes a model for unknown $\pi_{c}(\mathbf{x}_{i}\mid\bm{\beta}_{c})$ that we wish to estimate. By a symmetric process for the reference sample we have, $p(\mathbf{x}_{i} \mid i \in S_{r}) = \pi_{r}(\mathbf{x}_{i}\mid\bm{\beta}_{r}) \times f(\mathbf{x}_{i}) / P(i \in S_{r} \mid i\in S)$.  

We note that \emph{both} specifications for conditional distributions for $\mathbf{x}_{i}$ in each sampling arm use the same marginal distribution, $f(\mathbf{x}_{i})$, because both samples are drawn from the same underlying population.

Let $q = P(i \in S_{c}) = \int \pi_{c}(\mathbf{x}_{i}\mid\bm{\beta}_{c})f(x)dx$ and $t = P(i \in S_{r}) = \int \pi_{r}(\mathbf{x}_{i}\mid\bm{\beta}_{r})f(x)dx$ denote the unknown marginal probabilities used above to specify the conditional distributions for $\mathbf{x}_{i}$ in each sampling arm.  The marginal (over predictors, $\mathbf{x}$) probability for a unit to be selected into a sampling arm is denoted by $h = P(z_{i} = 1) = P(i \in S_{c}\mid i\in S)$. All of $(h,q,t,\bm{\beta}_{c},\bm{\beta}_{r})$ are unknown parameters that will receive prior distributions to be updated by the data.

The conditional distributions for $\mathbf{x}_{i}$ in each arm and the marginal probabilities for selection into each arm parameterize the likelihood for $(h,q,t,\bm{\beta}_{c},\bm{\beta}_{r})$,
\begin{equation}
\begin{split}
L\left(h,q,t,\bm{\beta}_{c},\bm{\beta}_{r}\mid \mathbf{z},X\right) \times \mathop{\prod}_{i=1}^{n}f(\mathbf{x}_{i}) =  \mathop{\prod}_{i=1}^{n}\left(h\pi_{c}(\mathbf{x}_{i}\mid\bm{\beta}_{c})/q\right)^{z_{i}} \times \left(((1-h)\pi_{r}(\mathbf{x}_{i}\mid\bm{\beta}_{r})/t\right)^{1-z_{i}} \\ \times f(\mathbf{x}_{i}),
\end{split}
\label{eq:rawlike}
\end{equation}
where we factor out the $f(\mathbf{x}_{i})$ on both sides and subsequently propose to drop these marginal distributions for the covariates because we don't believe they are random.  We use $f(\mathbf{x}_{i})$ as a computation device to allow us to specify a likelihood with conditional distributions of $\mathbf{x}_{i}$ in each sampling arm.

We proceed to reparameterize Equation~\ref{eq:rawlike} by extending an approach of \citet{johnson2021} from the case-control setting to our two-arm sampling set-up.   We construct the following transformed parameters:
\begin{align}
    \begin{split}
        \psi &= q(1-h)/th\\
        \pi_{zi}  &= \ \pi_c(\mathbf{x}_{i}\mid\bm{\beta}_c)/(\pi_c(\mathbf{x}_{i}\mid\bm{\beta}_c) + \psi \pi_r(\mathbf{x}_{i}\mid\bm{\beta}_r))\\
        1-\tilde{q}_{i} & = (1-h)\pi_r(\mathbf{x}_{i}\mid\bm{\beta}_r)/t.
    \end{split}
    \label{eq:transparms}
\end{align}

Using the transformations of Equation~\ref{eq:transparms} allows us to reparameterize the conditional likelihood (after dropping $f(\mathbf{x}_{i})$ in Equation~\ref{eq:rawlike}) to,
\begin{equation}
 L\left(h,q,t,\bm{\beta}_{c},\bm{\beta}_{r}\mid \mathbf{z},X\right) = \mathop{\prod}_{i=1}^{n} \pi_{zi}^{z_{i}}(1-\pi_{zi})^{1-z_{i}} \times \frac{1-\tilde{q}_{i}}{1-\pi_{zi}},
\end{equation}
which is a product of a Bernoulli distributed term and a ratio of transformed parameters.

We examine the non-Bernoulli likelihood contribution, $\mathop{\prod}_{i=1}^{n}\frac{1-\tilde{q}_{i}}{1-\pi_{zi}}$ asymptotically as the reference sample size, $n_r \uparrow \infty$, under a fixed convenience sample size, $n_c$.   We present a theoretical result in the following section that demonstrates the log of this ratio contribution to the likelihood limits to $0$, asymptotically for $n_r$ and $n_r/n_c$ both sufficiently large to allow the ignoring or dropping of this term.  


We may construct a model using the Bernoulli likelihood,
\begin{equation}
 L\left(h,q,t,\bm{\beta}_{c},\bm{\beta}_{r}\mid \mathbf{z},X\right) = \mathop{\prod}_{i=1}^{n} \pi_{zi}^{z_{i}}(1-\pi_{zi})^{1-z_{i}},
\end{equation}

with associated propensity,
\begin{equation}
    \pi_{zi}  = \ \pi_c(\mathbf{x}_{i})/(\pi_c(\mathbf{x}_{i}) + \psi \pi_r(\mathbf{x}_{i}))\
\end{equation}
where we have suppressed $(\bm{\beta}_{c},\bm{\beta}_{r})$ to facilitate comparison with $\pi_{z}(\mathbf{x}_{i}) = \pi_c(\mathbf{x}_{i})/(\pi_c(\mathbf{x}_{i}) + \pi_r(\mathbf{x}_{i}))$ from our main method.  

We see that the propensity formulation here and under our main method are nearly identical, up to an inclusion of $\psi$ in the denominator under the Symmetric two-arm approach, despite both being derived from different principles.  We prove in the next section that under the above definitions for $(h,q,t)$ that $\psi$ must equal $1$, which may be seen intuitively by noting that for a sample size, $n$, sufficiently large we may plug in modal quantities, $(h = n_{c}/n, q = n_{c}/N, t = n_{r}/N)$, for those marginal probabilities which produces $\psi = 1$.  Our proof for $\psi=1$ is true, however, for any sample size. The implication is that we have arrived at the very same result for the likelihood and conditional propensity, $\pi_{zi}$, as developed under our main approach.  The reverse implication is that the classical setting for \cite{lancaster1996case} could be estimated more efficiently by setting $\psi = 1$. Investigating whether this simplification for $\psi = 1$ holds for more complex applications such as $k$-indexed simultaneous outcomes with uniques values for $\psi_k$ \citep{johnson2021} is a subject for future work.



\subsection{Proof that $\log\left(\frac{1-\tilde{q}_{i}}{1-\pi_{zi}}\right)$  asymptotically contracts on $0$.}
\label{a:proof}

This proof performs an extension to the corresponding proof for stratified use-availability designs found in \cite{johnson2021} to our case of a the two-arm sampling design under an arbitrary sampling design.

\begin{proposition}\label{p:x}
	The pseudo log likelihood contribution $\sum_{i=1}^{n}\log\left(\frac{1-\tilde{q}_{i}}{1-p_{zi}}\right)$ contracts on  $0$ as the reference sample grows, $n_{r} \uparrow \infty$ and $h = n_c/n \downarrow 0$.
\end{proposition}
\begin{proof}
    We begin with some simple algebra to state the likelihood term with marginal probabilities, $(h,q,t)$,
    \begin{align}
    \begin{split}
     \mathop{\prod}_{i=1}^{n}\frac{1-\tilde{q}_{i}}{1-\pi_{zi}}
     &= \mathop{\prod}_{i=1}^{n}\frac{(1-h)\pi_{ri}}{t} \times \frac{\pi_{ci} + \psi\pi_{ri}}{\psi\pi_{ri}}\\
     &= \mathop{\prod}_{i=1}^{n}(1-h) \times \left(\frac{\pi_{ci}}{\psi t} + \frac{\pi_{ri}}{t}\right)\\
     &= \mathop{\prod}_{i=1}^{n}(1-h) \times \left(\frac{\pi_{ci}}{\psi_{c}} + \frac{\pi_{ri}}{t}\right)
     \end{split}
     \label{eq:mult}
    \end{align}
    where for readability we simplify the expression of $\pi_{c}(\mathbf{x}_{i}\mid\bm{\beta}_{c})$ with the short-hand, $\pi_{ci}$, and the same for $\pi_{ri}$.  We plug in for $\psi = \frac{q(1-h)}{h} \times \frac{1}{t} = \frac{\psi_{c}}{t}$ into the last equation in the series where $\psi_{c}$ is composed of quantities solely related to the convenience sample.
    
    We take the logarithm of the last equation of Equation~\ref{eq:mult},
    \begin{equation}
       \log\left(\mathop{\prod}_{i=1}^{n}(1-h) \times \left(\frac{\pi_{ci}}{\psi_{c}} + \frac{\pi_{ri}}{t}\right)\right) =  n\log(1-h) + \mathop{\sum}_{i=1}^{n}\log\left(\frac{\pi_{ci}}{\psi_{c}} + \frac{\pi_{ri}}{t}\right).
    \end{equation}
    We proceed to take a Taylor series expansion of $\log\left(\frac{\pi_{ci}}{\psi_{c}} + \frac{\pi_{ri}}{t}\right)$ about $\frac{\pi_{ci}}{\psi_{c}} = 0$ and use the first term, which we may do since $\frac{\pi_{ci}}{\psi_{c}}$ grows vanishingly small in the limit as $n\uparrow\infty$ (since $h \downarrow 0$ such that $\psi_{c}\uparrow\infty$).  This produces,
    
    \begin{align}
      n\log(1-h) + \mathop{\sum}_{i=1}^{n}\log\left(\frac{\pi_{ci}}{\psi_{c}} + \frac{\pi_{ri}}{t}\right) &= n\log(1-h) +\mathop{\sum}_{i=1}^{n}\frac{\pi_{ci}}{\psi_{c}}\frac{t}{\pi_{ri}}  \\
      &=n\log(1-h) +\frac{t}{\psi_{c}}\mathop{\sum}_{i=1}^{n}\frac{\pi_{ci}}{\pi_{ri}}\\
      &=n\log(1-h) + \frac{t}{\psi_{c}} \times \overline{\left[\frac{\pi_{c}}{\pi_{r}}\right]}\\
      &=n\log(n-n_c) - n\log n +\frac{n_c}{n-n_{c}}\times \frac{nt}{q}\overline{\left[\frac{\pi_{c}}{\pi_{r}}\right]},
    \end{align}
    where we have plugged in $h = n_c/n$ for $n$ sufficiently large and $\overline{\left[\frac{\pi_{c}}{\pi_{r}}\right]}$ represents the mean of the ratio,  $\frac{1}{n}\mathop{\sum}_{i=1}^{n}\frac{\pi_{ci}}{\pi_{ri}}$.
    
    We next evaluate the limit of the above expression as $n_r \uparrow \infty$ under a constant value for $n_c$,
    \begin{subequations}
    \begin{align}
        \mathop{\lim}_{n_{r}\uparrow\infty} n\log(n-n_c) - n\log n +\frac{n_c}{n-n_{c}}\times \frac{nt}{q}\overline{\left[\frac{\pi_{c}}{\pi_{r}}\right]} &= -n_{c} + n_{c}\frac{t}{q}\overline{\left[\frac{\pi_{c}}{\pi_{r}}\right]}\\
         &= -n_{c} + n_{c}\frac{t}{q}\frac{\overline{\pi_c}}{\overline{\pi_r}}\label{eq:ratiosum}\\
         & = -n_{c} + n_{c}\frac{t}{q}\frac{q}{t}\label{eq:LLN}\\
         & = -n_{c} + n_{c}.
    \end{align}
    \end{subequations}
    
    In Equation~\ref{eq:ratiosum}, the mean of the ratios contracts on the ratio of the means as $n\uparrow\infty$ because $\pi_{ci}$ limits to $0$ as $n$ increases since $n_c$ is fixed such that the $\lim_{n\uparrow\infty}\frac{\pi_{cn}}{\pi_{rn}}$ exists and is finite.  Also required is that $\pi_{rn}$ be non-decreasing as $n$ increases, which we may achieve through reordering the terms.  Next, we apply the Law of Large Numbers for the convergence of the sample mean under the assumption of absolutely bounded values in expectation for $q$ and $t$.
    
 \end{proof}   
 
\subsection{Proof that $\psi = 1$ under the Symmetric Two-arm Method of Section~\ref{sec:symmetric}}
\label{b:proof}
\begin{proposition}\label{p:psi}
	Let marginal (over $\mathbf{x}_{i}$) probabilities be defined as,  $h = P(i \in S_{c}\mid i \in S),~ q = P(i \in S_{c} \mid i \in U),~ t = P(i \in S_{r} \mid i \in U)$ and further define $\psi = q(1-h)/th$.  Let $\mathcal{S}$ denote the space of all two-arm samples, $(S_{c},S_{r})$ of size $(n_{c},n_{r})$, respectively. Recall that $U$ is the set of two stacked populations $\{U^{0}, U^{0}\}$ corresponding to each arm.
	\noindent Then if we assume strictly positive conditional inclusion probabilities for all units and that the convenience sample arises from an underlying latent random sampling design then, 
	\vspace{0.5cm}
	\newline $\psi = 1$ a.s. $P_{\pi}$, where $P_{\pi}$ is the unknown true joint generating distribution for all $(S_{c},S_{r}) \in U \subset \mathcal{S}$, given $U$.
\end{proposition}
\begin{proof}
For any $S = S_{c}~+~S_{r} \in U \subset \mathcal{S}$, 

\begin{subequations}
\begin{align}
    \psi &= \frac{q(1-h)}{th}\\
    \psi\frac{h}{q} &= \frac{1-h}{t}\\
    \psi\frac{P(i \in S_{c}\mid i \in S)}{P(i \in S_{c} \mid i \in U)} &= \frac{P(i \in S_{r}\mid i \in S)}{P(i \in S_{r} \mid i \in U)}\label{eq:base}\\
    \psi\frac{P(i \in S_{c}\mid i \in S)P(i\in S \mid i \in U)}{P(i \in S_{c} \mid i \in U)} &= \frac{P(i \in S_{r}\mid i \in S)P(i\in S \mid i \in U)}{P(i \in S_{r} \mid i \in U)}\label{eq:multiply}\\
    \psi P(i \in S \mid i \in S_{c}) &= P(i \in S \mid i \in S_{r})\\
    \psi &= 1,
\end{align}
\end{subequations}

\noindent where in Equation~\ref{eq:multiply} we multiply both left- and right-hand side of Equation~\ref{eq:base} by $P(i \in S \mid i \in U) > 0$.

\end{proof}
\emph{Remark:} 
When the reference sample is the population $S_{r} = U^{0}$, the proof holds without modification.


\section{Stan Model Estimation Script}\label{a:script}

We present the Stan estimation script \citep{gelman2015a} for our two-arm exact likelihood method, below.  The script is built around Stan's \texttt{partial\_sum} function to allow within chain parallelization for computational scalability.

\begin{lstlisting}
functions{
  
  vector build_b_spline(vector t, vector ext_knots, int ind, int order);
  matrix build_mux(int N, int start, int end, int K_sp, matrix X, matrix[] G, 
                   matrix beta_x, matrix[] beta_w);
  row_vector build_muxi(int K_sp, int num_basis, row_vector x_i, vector[] g_i, 
                        matrix beta_x, matrix[] beta_w);
  vector build_b_spline(vector t, vector ext_knots, int ind, int order) {
    // INPUTS:
    //    t:          the points at which the b_spline is calculated
    //    ext_knots:  the set of extended knots
    //    ind:        the index of the b_spline
    //    order:      the order of the b-spline
    vector[num_elements(t)] b_spline;
    vector[num_elements(t)] w1 = rep_vector(0, num_elements(t));
    vector[num_elements(t)] w2 = rep_vector(0, num_elements(t));
    if (order==1)
      for (i in 1:num_elements(t)) // B-splines of order 1 are piece-wise constant
        b_spline[i] = (ext_knots[ind] <= t[i]) && (t[i] < ext_knots[ind+1]);
    else {
      if (ext_knots[ind] != ext_knots[ind+order-1])
        w1 = (to_vector(t) - rep_vector(ext_knots[ind], num_elements(t))) /
             (ext_knots[ind+order-1] - ext_knots[ind]);
      if (ext_knots[ind+1] != ext_knots[ind+order])
        w2 = 1 - (to_vector(t) - rep_vector(ext_knots[ind+1], num_elements(t))) /
                 (ext_knots[ind+order] - ext_knots[ind+1]);
      // Calculating the B-spline recursively as linear interpolation of two lower-order splines
      b_spline = w1 .* build_b_spline(t, ext_knots, ind, order-1) +
                 w2 .* build_b_spline(t, ext_knots, ind+1, order-1);
    }
    return b_spline;
  }
  
  matrix build_mux(int N, int start, int end, int K_sp, matrix X, matrix[] G, matrix beta_x, matrix[] beta_w){
    matrix[N,2] mu_x;
    for( arm in 1:2 )
    {
      mu_x[1:N,arm]     = X[start:end,] * to_vector(beta_x[,arm]); /* N x l for each arm */
      // spline term
      for( k in 1:K_sp )
      {
        mu_x[1:N,arm]  += to_vector(beta_w[arm][,k]' * G[k][,start:end]); /* N x 1 */
      } /* end loop k over K predictors */
      
    }/* end loop arm over convenience and reference sample arms */
    
    return mu_x;
  }
  
  row_vector build_muxi(int K_sp, int num_basis,
     row_vector x_i, vector[] g_i, matrix beta_x, matrix[] beta_w){
       
     row_vector[2] mu_xi;
    
     for( arm in 1:2 )
     {
       mu_xi[arm]     = dot_product(x_i,beta_x[,arm]); /* scalar */
       // spline term
       for( k in 1:K_sp )
       {
         mu_xi[arm]   += dot_product(beta_w[arm][1:num_basis,k], g_i[k][1:num_basis]); /* scalar */
       } /* end loop k over K predictors */
      
     }/* end loop arm over convenience and reference sample arms */
    
    return mu_xi;
    
  } /* end function build_mu */
  
  real partial_sum(int[] s, 
                   int start, int end, real[] logit_pw, int K_sp, int n_c, int n,
                   int num_basis, matrix X, matrix[] G, matrix beta_x, matrix[] beta_w, 
                   real phi_w) {
     int N = end - start + 1;
     matrix[N,2] mu_x;
     matrix[N,2] p;
     vector[N] p_tilde; // this pseudoprobability must be in [0,1]
     real fred              = 0;
     
     // memo: slicing on all of mu_x[li,arm], p[li,arm], p_tilde[li] for li in 1:(end-start+1)
     //       where p_tilde is the mean vector for binary data vector, s, and mu_x[,2]
     //       is the mean vector for data vector logit_pw.
     //       Also slicing data vectors s and logit_pw in their respective
     //       log-likelihood contributions.
     
     mu_x             = build_mux(N, start, end, K_sp, X, G, beta_x, beta_w);
               
     p                = inv_logit( mu_x );
    
    
    // 1. bernoulli likelihood contribution
    p_tilde         = p[1:N,1] ./ ( p[1:N,1] + p[1:N,2] );
    fred            += bernoulli_lpmf( s[1:N] | p_tilde );
    
    // 2. normal likelihood contribution
    // In non-threaded model, likelihood statement for n - n_c, logit_pw
    // logit_pw          ~ normal( mu_x[(n_c+1):n,2], phi_w ); 
    // slicing on n length vector mu_x[,2] 
    // subsetting portion of mu_x[,2] linked to logit_pw
    if( start > n_c  ) ## all units in this chunk increment likelihood contribution for logit_pw
    {
      fred       += normal_lpdf( logit_pw[start-n_c:end-n_c] | mu_x[1:N,2], phi_w );
    }else{ /* start <= n_c */
      if( end > n_c ) /* some units in chunk < n_c and some > n_c; only those > n_c increment likelihood */
      {
        fred                      += normal_lpdf( logit_pw[1:end-n_c] | mu_x[n_c-start+2:N,2], phi_w );
        
      } /* end if statement on whether n_c \in (start,end ) */
    } /* end if-else statement on whether add logit_pw likelihood contributions */
    
  
  return fred;
  }/* end function partial_sum() */

} /* end function block */

data{
    int<lower=1> n_c; // observed convenience (non-probability) sample size
	  int<lower=1> n_r; // observed reference (probability) sample size
	  int<lower=1> N; // estimate of population size underlying reference and convenience samples
	  int<lower=1> n; // total sample size, n = n_c + n_r
	  int<lower=1> num_cores;
	  int<lower=1> multiplier;
	  int<lower=1> K; // number of fixed effects
	  int<lower=0> K_sp; // number of predictors to model under a spline basis
	  int<lower=1> num_knots;
    int<lower=1> spline_degree;
    matrix[num_knots,K_sp] knots;
    real<lower=0> weights[n_r]; // sampling weights for n_r observed reference sample units
    matrix[n_c, K] X_c; //  *All* predictors - continuous and categorical -  for the convenience units
    matrix[n_r, K] X_r; // *All* predictors - continuous and categorical -  for the reference units
    matrix[n_c, K_sp] Xsp_c; //  *Continuous* predictors under a spline basis for convenience units
    matrix[n_r, K_sp] Xsp_r; // *Continuous* predictors under a spline basis for convenience units
    int<lower=1> n_df;
} /* end data block */

transformed data{
  // create indicator variable of membership in convenience or reference samples
  // indicator of observation membership in the convenience sample
  int grainsize                     = ( n / num_cores ) / multiplier;
  real logit_pw[n_r]                = logit(inv(weights));
  int<lower=0, upper = 1> s[n]      = to_array_1d( append_array(rep_array(1,n_c),rep_array(0,n_r)) ); 
  matrix[n,K] X                     = append_row( X_c,X_r );
  matrix[n,K_sp] X_sp               = append_row( Xsp_c,Xsp_r );
  /* formulate spline basis matrix, B */
  int num_basis = num_knots + spline_degree - 1; // total number of B-splines
  matrix[spline_degree + num_knots,K_sp] ext_knots_temp;
  matrix[2*spline_degree + num_knots,K_sp] ext_knots;
  matrix[num_basis,n] G[K_sp]; /* basis for model on p_c */
  for(k in 1:K_sp)
  {
     ext_knots_temp[,k] = append_row(rep_vector(knots[1,k], spline_degree), knots[,k]);
    // set of extended knots
     ext_knots[,k] = append_row(ext_knots_temp[,k], rep_vector(knots[num_knots,k], spline_degree));
     for (ind in 1:num_basis)
     {
        G[k][ind,] = to_row_vector(build_b_spline(X_sp[,k], ext_knots[,k], ind, (spline_degree + 1)));
     }
     G[k][num_knots + spline_degree - 1, n] = 1;
  }
     
} /* end transformed data block */

parameters {
  matrix<lower=0>[K,2] sigma2_betax; /* standard deviations of K x 2, beta_x */
                                    /* first column is convenience sample, "c", and second column is "r" */

  matrix[K,2] betaraw_x; /* fixed effects coefficients - first colum for p_c; second column for p_r  */
  // spline coefficients
  vector<lower=0>[2] sigma2_global; /* set this equal to 1 if having estimation difficulties */
  matrix<lower=0>[2,K_sp] sigma2_w;
  matrix[num_basis,K_sp] betaraw_w[2];  // vector of B-spline regression coefficients for each predictor, k
                                        // and 2 sample arms
  real<lower=0> phi2_w; /* scale parameter in model for -1og(weights) */                                     
} /* end parameters block */

transformed parameters {
  matrix[K,2] beta_x;
  matrix[num_basis,K_sp] beta_w[2];
  matrix<lower=0>[K,2] sigma_betax;
  vector<lower=0>[2] sigma_global; /* set this equal to 1 if having estimation difficulties */
  matrix<lower=0>[2,K_sp] sigma_w;
  real<lower=0> phi_w;
  
  sigma_betax       = sqrt( sigma2_betax );
  sigma_global      = sqrt( sigma2_global );
  sigma_w           = sqrt( sigma2_w );
  phi_w             = sqrt( phi2_w );
  
  // for scale parameters for interaction effects from those for main effects to which they link
  for( arm in 1:2 )
  {
     beta_x[,arm]   = betaraw_x[,arm] .* sigma_betax[,arm]; /* Non-central parameterization */
  }/* end loop arm over convenience and reference sample arms */
  
  // spline regression coefficients
  for(arm in 1:2)
  {
    for( k in 1:K_sp ) 
    {
      beta_w[arm][,k]   = cumulative_sum(betaraw_w[arm][,k]);
      beta_w[arm][,k]  *= sigma_w[arm,k] * sigma_global[arm];
    } /* end loop k over K predictors */
  }
  
} /* end transformed parameters block */

 model {
  to_vector(sigma2_betax) ~ gamma(1,1);
  to_vector(sigma2_w)     ~ gamma(1,1);
  sigma_global            ~ gamma(1,1);
  phi2_w                  ~ gamma(1,1);

  to_vector(betaraw_x)   ~ std_normal();
  for(arm in 1:2)
    to_vector(betaraw_w[arm])   ~ std_normal();
  
  /* Model likelihood for y, logit_pw */
 // Sum terms 1 to n in the likelihood 
 target            += reduce_sum(partial_sum, s, grainsize, 
                                              logit_pw, K_sp, n_c, n, num_basis, X, G, 
                                              beta_x, beta_w, phi_w); 

} /* end model block */  

generated quantities{
  matrix[n,2] p;
  matrix[n,2] mu_x;
  
  for( arm in 1:2 )
  {
    mu_x[,arm]     = X[,] * to_vector(beta_x[,arm]); /* n x l for each arm */
    // spline term
    for( k in 1:K_sp )
    {
      mu_x[,arm]  += to_vector(beta_w[arm][,k]' * G[k][,1:n]); /* n x 1 */
    } /* end loop k over K predictors */
  
   }/* end loop arm over convenience and reference sample arms */
  
  p        = inv_logit( mu_x );
  
  // smoothed sampling weights for convenience and reference units
  vector[n] weights_smooth_c = inv(p[,1]);
  vector[n] weights_smooth_r = inv(p[,2]);
  // inclusion probabilities in convenience and reference units for convenience units
  // use for soft thresholding
  vector[n_c] pi_c                  = p[1:n_c,1];
  vector[n_c] pi_r_c                = p[1:n_c,2];
  // normalized weights
  weights_smooth_c                  *= ((n_c+0.0)/(n+0.0)) * (sum(weights_smooth_r)/sum(weights_smooth_c));
  weights_smooth_r                  *= ((n_r+0.0)/(n+0.0));
} /* end generated quantities block */
\end{lstlisting}

\section{Simulation Results for Estimating Population Mean, $\mu$}\label{a:mu}
We use the convenience sample inclusion probabilities, $\pi_{cmi},~i \in S_{c}$, estimated from models on each Monte Carlo iteration, $m = 1,\ldots,(M=30)$, to form a population mean estimator, $\mu_{m}$. As discussed in the introduction, we use our Bayesian hierarchical model to estimate $\pi_{cmi}$, such these latent sampling probabilities may be used to construct survey estimators for focus response variables.   We construct $\mu_{m} = \frac{\sum_{i \in S_{c}} y_i/\hat{\pi}_{cmi} + \sum_{j \in S_{r}} y_j/\pi_{rmj}}{\sum_{i \in S_{c}} 1/\hat{\pi}_{cmi} + \sum_{j \in S_{r}} 1/\pi_{rmj}}$ as a sample-weighted (Hajek) survey direct estimator, so there is no additional model specified; that is, the estimator each $\mu_{m}$ assumes the underlying population values for $y_{mi}$ are fixed such the estimator is random with respect to the distribution that governs the taking of samples from that fixed population.

We propagate uncertainty in the model-based estimation of the convenience sample inclusion probabilities by taking multiple draws or imputes (e.g., $J = 10$) of each inclusion probability from its posterior distribution.  We formulate the survey direct estimator using that draw of the inclusion probabilities.   We compute the variance of the survey estimator for the population mean with respect to the survey sampling distribution.  We repeat this procedure for each draw and then compute the between draws variance variance with respect to the modeling distribution.  We put it together by using multiple imputation combining rules to construct a total variance for our survey direct estimate that now incorporates uncertainty with respect to both the model for estimating inclusion probabilities and the distribution governing the taking of samples.

More specifically, we construct a total, combined variance of the form $T = (1+ 1/J)B + \bar{U}$ based on multiple imputation rules of \citet[][See section 2.1.1]{Reiter2007}, where $T$ denotes the total variance of our $\mu$ estimator that accounts for both uncertainty with respect to drawing samples and with respect to the modeling of the inclusion probabilities used to form sampling weights.   Let $j \in 1,\ldots,M$ index a randomly drawn imputation for $(\hat{\pi}_{cji})_{i\in S_{c}}$,  $(\hat{\pi}_{rji^{'}})_{i^{'}\in S_{r}}$ from the set of MCMC samples for a model run. $\bar{U}$ denotes the within imputation sampling variance of $\mu_{j}$ and $B$ denotes the between modeled variance of $\mu_{j}$ across the $J$ imputations.

Once the total variance is computed, we then generate symmetric asymptotic intervals using the $t-$ distribution.  The use of multiple imputation allows us to propagate the uncertainty in estimation of $\pi_{cji}$ into the variance estimate for our direct estimator of $\mu_{j}$.   We next present details to construct the within impute variance, $\bar{U}$, and the between impute variance, $B$.

In what follows, we assume that we use the model-smoothed estimator, $\hat{\pi}_{rji} = \mu_{x,rji}$ (from Equation~14) for the reference sample inclusion probabilities to construct the mean statistic.   We compare simulation study results for $\mu$ using both using the fixed $\pi_{rji}$ and the model-smoothed $\hat{\pi}_{rji}$ in the sequel.

\citet{binder1996taylor} provides a general approach to Taylor linearization for computing the within impute variance. We fix an imputation $j \in (1,\ldots,M)$. For a simple weighted linear statistic such as $\mu_{j}$, the approach simplifies to calculating the variance of the weighted residuals $w_{ji} (y_i -  \mu_{j})$ with weights $w_{ji} =  \hat{\pi}_{cji} / \left(\sum_{i \in S_{c}} 1/\hat{\pi}_{cji} + \sum_{i^{'} \in S_{r}} 1/\hat{\pi}_{rji^{'}} \right)$ for convenience sample units or \newline
$w_{ji^{'}} =  \hat{\pi}_{rji^{'}} / \left(\sum_{i \in S_{c}} 1/\hat{\pi}_{cji} + \sum_{i^{'} \in S_{r}} 1/\hat{\pi}_{rji^{'}} \right)$ for reference sample units.  We average over the $J$ within-impute design (sample)-based variance estimates of $\mu_{j}$ (via Taylor linearization) to get $\bar{U}$. 

We proceed to construct the model-based, between variance $B$ by computing the variance over the $J$ imputations for $\mu_{j}$.

We first illustrate the benefit of incorporating the sample weighted convenience units with the reference sample units into the computation of $\mu$.  We then proceed to compare the pseudo likelihood methods for $\pi_{ci}$ with our two-arm exact likelihood method under combined reference and convenience sample estimation of $\mu$.   

Finally, the two-arm method co-estimates $\pi_{ri},~i \in S_{c}$ simultaneously with estimating $\pi_{ci}$.  So, on each Monte Carlo iteration, $m$, we threshold or exclude those convenience sample units, $\{\ell \in S_{c} : \pi_{rm\ell} < \epsilon\}$; that is, we exclude those convenience units that express small reference sample inclusion probabilities in order to reduce noisiness in our estimator.   We experiment with setting $\epsilon = (Q_{1},Q_{5},Q_{10})$, where $Q_p$ is the $p^{th}$ percent quantile of the distribution of smoothed $\pi_{ri},~i \in S_{r}$.

Results are presented in the Figures~\ref{fig:con_v_inf} - \ref{fig:two_arm_variations}.  Each plot panel from left to right measures the bias, root mean squared error, mean absolute deviation and coverage for estimated $\mu$.

We construct separate convenience samples under both the low- and high-overlap sampling designs used in the previous results for estimating the conditional convenience sample inclusion probabilities.  In each row of every plot panel we present the result for the low-overlap sampling design, labeled ``L" and the high-overlap sampling design, labeled "H" with those results connected by a horizontal bar.  In practice, a convenience dataset will probably lie somewhere in between L and H.

Lastly, we lay in the result for the base case that constructs $\mu$ solely from the reference sample as a dashed black vertical line in each plot panel in all of the figures.

\afterpage{\FloatBarrier}

Figure~\ref{fig:con_v_inf} compares constructing $\mu$ solely from the convenience sample in the first row to using both the reference and convenience samples (both under true inclusion weights) in the second row.  We see a dramatic improvement in the quality of estimated $\mu$ under the high-overlap convenience samples and a smaller, but still notable improvement under the low-overlap convenience samples, on the one hand, from use of solely the convenience samples, on the other hand.   

\begin{figure}
\centering
\includegraphics[width = 0.95\textwidth]{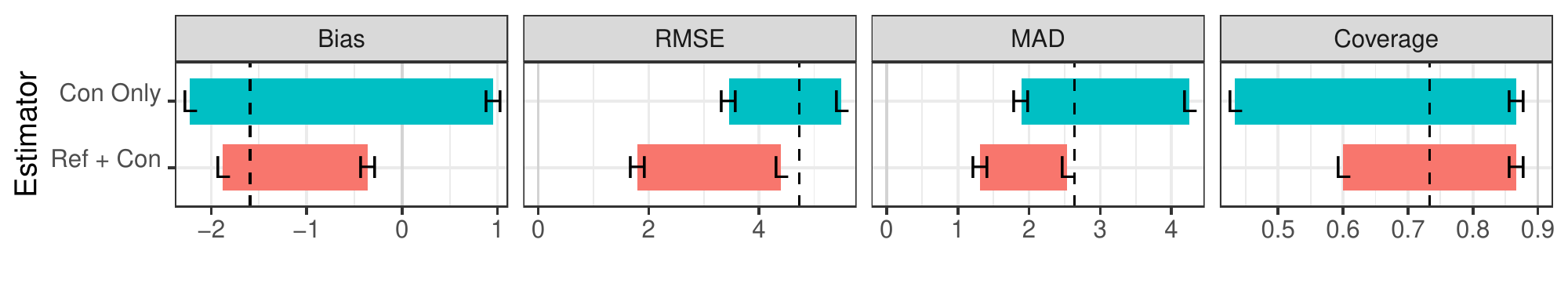}
\caption{Performance of the weighted mean estimator between high (H) and low (L) overlapping samples using the convenience and reference sample with true weights across Monte Carlo Simulations for (top to bottom) Only Convenience, Convenience and Reference Sample. Left to right: Bias, root mean square error, mean absolute deviation, coverage of 90\% intervals. Vertical reference line corresponds to using the reference sample only.
    }
\label{fig:con_v_inf}
\end{figure}
\afterpage{\FloatBarrier}

Figure~\ref{fig:main_methods_compare} compares the quality of estimated $\mu$ between our exact two-arm method (using published / known reference sample inclusion probabilities) with the pseudo likelihood methods. The first row presents the combined reference and convenience sample using the true values for the latent convenience sample weights as a comparator for all methods.   The second row presents the combined reference and convenience samples now using the estimated convenience sample inclusion weights under our two-arm method.   The performance is very similar to using the true inclusion weights for the convenience sample.   The third row presents the CLW method of \citet{doi:10.1080/01621459.2019.1677241}, which performs relatively poorly due to high estimation variation expressed by this method in estimation of $\pi_{cmi}$ for our moderate sample sizes $(n_r,n_c) = (400,\sim 800)$.   We achieve best performance for high-overlap convenience samples (labeled (H)).  The last row presents the same, but using the WVL pseudo likelihood method of \citet{2021valliant} that expresses less variation in estimation of $\pi_{cmi}$ than does CLW (though still substantially higher than our two-arm method).  Yet, even though the estimated weights under WVL are biased under both low- and high-overlap samples, the method performs similarly in estimation of $\mu$ to our two-arm method because the bias for WVL is largest at high values for $\pi_{ci}$ while most sampled units are assigned $\pi_{ci} < 0.75$.   The low-overlap samples produce notably worse coverage under WVL due to the bias and failure to account for uncertainty in $\pi_{rmi}$ by using them as plug-in.

It bears mention that even when using the true values for $\pi_{ci}$ the coverage of the estimator for $\mu$ under the low-overlap datasets fails to achieve nominal coverage because of our moderate sample sizes.  These moderate sample sizes realistically reflect the sampling of domains (e.g., geographic-by-industry for establishment surveys) used in practice.   We render an estimator using the true sampling weights in each plot panel so that we may understand the performance of the methods in context of the best possible performance.

\begin{figure}
\centering
\includegraphics[width = 0.95\textwidth]{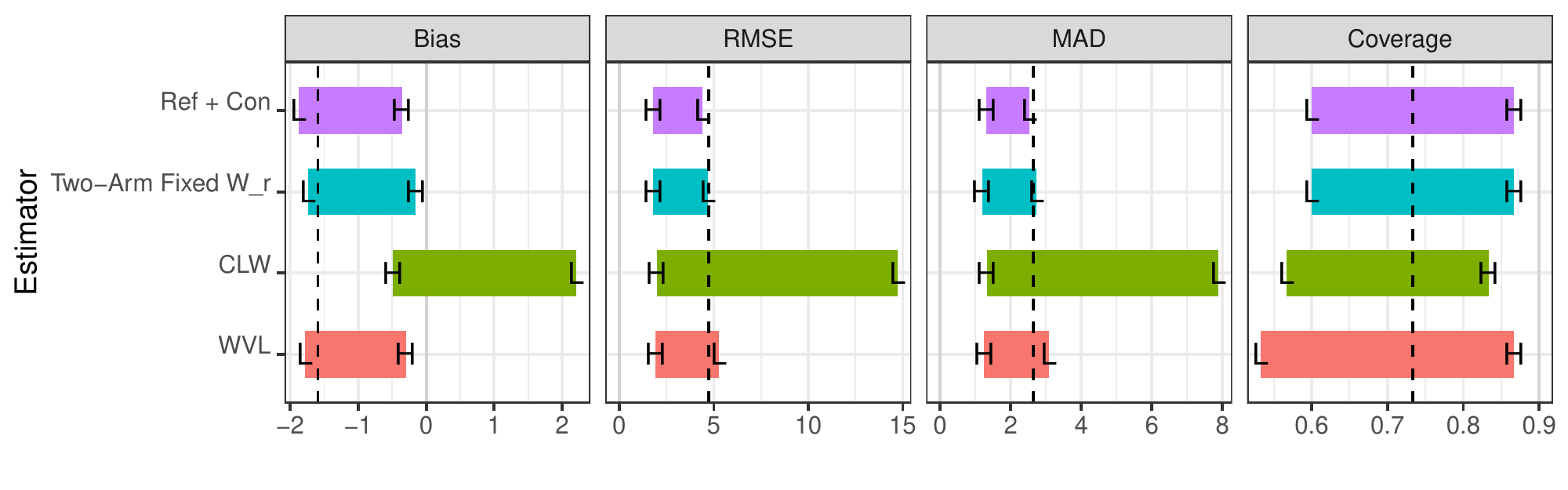}
\caption{Performance of the weighted mean estimator between high (H) and low (L) overlapping samples using the convenience with modeled weights and reference sample with true weights across Monte Carlo Simulations for (top to bottom) True weights, Two-Arm weights, CLW weights, WVL weights. Left to right: Bias, root mean square error, mean absolute deviation, coverage of 90\% intervals. Vertical reference line corresponds to using the reference sample only.
    }
\label{fig:main_methods_compare}
\end{figure}
\afterpage{\FloatBarrier}

We conclude the exploration of methods for estimating $\pi_{ci}$ on the quality of the resultant mean estimator, $\mu$, by comparing versions or variations of the two-arm method for estimating $\pi_{ci}$. So far we have seen that the two-arm method outperforms the other methods and, in particular, the pseudo likelihood methods for estimation of $\mu$ in terms of bias, means squared error and converge (uncertainty quantification).   

Since the two-arm method co-models $\pi_{ri}$ to borrow strength in estimation of $\pi_{ci}$, we may use either fixed $\pi_{ri}$ or modeled / smoothed values for $\pi_{ri}$ to form our combined reference and convenience estimator for $\mu$. The first row of Figure~\ref{fig:two_arm_variations} presents the combined reference and convenience-based estimator for $\mu$ that uses true sample weights as the benchmark comparator.  The second row uses fixed (and published) $\pi_{ri}$ along with our two-arm method for estimating $\pi_{ci}$ to produce the combined reference and convenience sample inclusion probability for $\mu$.  The third row is the same as the second except that we replace the fixed $\pi_{ri}$ with modeled or smoothed values from our two-arm model.
We observe that using smoothed weights improves the coverage for high-overlap datasets because co-modeling the $\pi_{ri}$ accounts for uncertainty in the generation of samples \citep{10.1214/19-EJS1538}.

Lastly, we next leverage our co-estimation of $\pi_{ri}$ for $i \in S_{c}$, which are typically \emph{unknown} for non-overlapping units between the two sample arms, to threshold inclusion of units from the convenience sample.   We seek to exclude those convenience sample units, $\ell \in S_{c}$ where the associated $\pi_{r\ell} < \epsilon$; that is, we exclude units from the convenience sample that are estimated with very small values for $\pi_{r\ell}$ in order to remove units that would induce noise in our estimator for $\mu$.   We see that the estimator for $\mu$ that results from setting $\epsilon = Q_{1}$ notably improves bias performance in the estimator for $\mu$ for low overlap samples while leading to only a slight increases in RMSE for high overlap.  When we increase to $\epsilon = Q_{5}$, bias increases slightly for the high-overlap datasets but further decreases for the low-overlap dataset.  The coverage performance, however, notably improves under $\epsilon = Q_{5}$ as compared to the non-thresholded two-arm-based estimator. The general pattern continues with $\epsilon = Q_{10}$.
This result suggests thresholding using $\epsilon \approx Q_{5}$ would be advisable, particularly for lower overlapping samples.

\begin{figure}
\centering
\includegraphics[width = 0.95\textwidth]{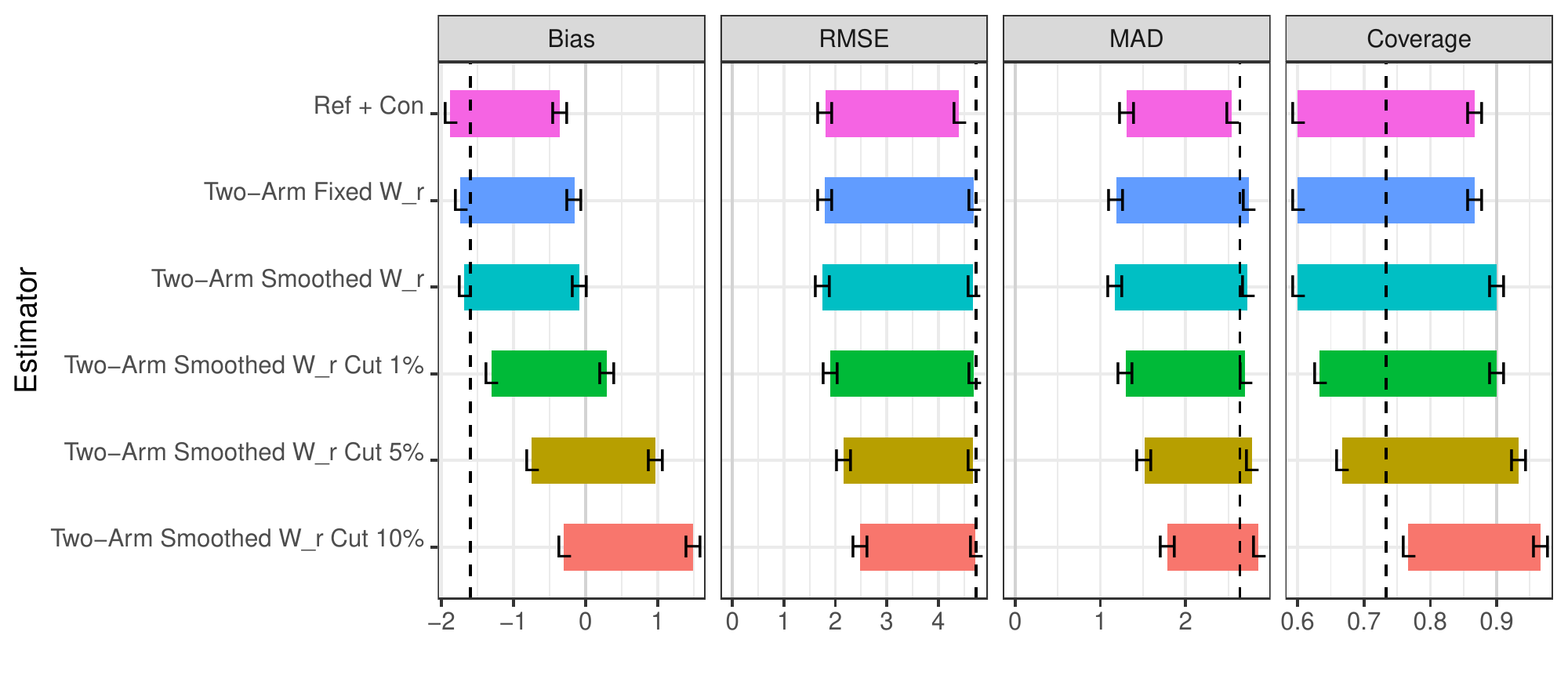}
\caption{Performance of the weighted mean estimator between high (H) and low (L) overlapping samples using variations of the two-arm method across Monte Carlo Simulations for (top to bottom) True weights for both samples, Original weights for Reference Sample, Smoothed weights for reference sample, Subset of convenience sample meeting 1\%, 5\%, and 10\% overlap threshold. Left to right: Bias, root mean square error, mean absolute deviation, coverage of 90\% intervals. Vertical reference line corresponds to using the reference sample only.
    }
\label{fig:two_arm_variations}
\end{figure}
\afterpage{\FloatBarrier}


\bibliography{ref}

\end{document}